\documentclass[a4paper,12pt]{article}
\usepackage{amsmath,amsthm,amsfonts,mathtools}
\usepackage{amssymb}
\usepackage{hyperref}
\usepackage[T1]{fontenc}
\usepackage{times}
\mathtoolsset{showonlyrefs=true}
\usepackage{paralist}
\usepackage{graphicx}

\setdefaultenum{i)}{(a)}{}{}
\usepackage[round]{natbib}

\usepackage[total={6in,10in},left=1.5in,top=0.5in,includehead,includefoot]{geometry}
\usepackage[dvipsnames,table]{xcolor}
\usepackage{subfig}
\usepackage{listings}

\definecolor{gris245}{RGB}{245,245,245}
\definecolor{olive}{RGB}{50,140,50}
\definecolor{brun}{RGB}{175,100,80}

\numberwithin{equation}{section}
\usepackage{enumerate} 
\usepackage{color}
\usepackage{endnotes}
\let\footnote=\endnote

\newtheorem{theorem}{Theorem}[section]

\newtheorem{definition}[theorem]{Definition}

\newtheorem{lemma}[theorem]{Lemma}
\newtheorem{proposition}[theorem]{Proposition}
\newtheorem{remark}[theorem]{Remark}

\numberwithin{equation}{section}

\newcommand{\nada}[1]{}

\DeclareMathOperator{\tracking}{\textup{TrE}}
\DeclareMathOperator{\avtrco}{\textup{ATC}}

\usepackage{csquotes}
\DeclareMathOperator{\tdiff}{\text{TrD}}

\long\def\symbolfootnote[#1]#2{\begingroup\def\thefootnote{\fnsymbol{footnote}}\footnote[#1]{#2}\endgroup}

\title{Asymptotic methods for transaction costs}
\author{Eberhard Mayerhofer\footnote{University of Limerick, Department of Mathematics \& Statistics, V94 T9PX, Limerick, Ireland, \texttt{eberhard.mayerhofer@ul.ie}.}}
\begin{document}
\maketitle

\abstract{
We propose a general approximation method for determining optimal trading strategies in markets with proportional transaction costs, with a polynomial approximation of the residual value function. The method is exemplified by several problems from optimally tracking benchmarks, hedging the Log contract, to maximizing utility from terminal wealth. Strategies are also approximated by practically executable, discrete trades. We identify the necessary trade-off between trading frequency and trade sizes to have satisfactory agreement with the theoretically optimal, continuous strategies of infinite activity.

\vspace{0.5cm}

\newpage
\section{Introduction}\label{sec: Intro}

\begin{quote}
{\it It was explained to me that the effects on which one was working were so vanishingly small that without the greatest possible precision in computation they might have been missed altogether.}
\flushright{Norbert Wiener\footnote{ Norbert Wiener about the Manhattan project, in ``I am a mathematician'', MIT Press 1964.}}
\end{quote}

Optimal trading strategies in a Black-Scholes market or its extensions face challenges when applied to realistic markets with frictions, including transaction costs, price impact, or margin requirements. For instance, employing a constant proportion trading strategy or a standard Delta-hedge in the Black-Scholes model becomes practically infeasible in the presence of a relative bid-ask spread $\varepsilon$, as it may lead to immediate bankruptcy due to trading policies of infinite variation. In order to maintain solvency in markets with proportional transaction costs, trading frequency is adjusted to finite variation, involving minimal trades necessary to stay close to the target exposures. The deviation from this target exposure is of order $\varepsilon^{1/3}$, while the impact on transaction costs is of second order (that is, of magnitude $\varepsilon^{2/3}$, cf.~\cite{MR1284980,gerhold.al.11,kallsenmad,rogers}).
Determining long-run optimal trading strategies or solutions to infinite horizon problems involves solving specific free boundary problems. These problems define boundaries that limit the maximum deviation of a portfolio statistic $\xi_t$ from the target, where these statistics commonly represent the proportion of wealth in the risky asset, the risky-safe ratio, or transformed variants. For constant investment opportunities, an investor refrains from trading when $\xi_t$ is outside a constant region $[\xi_-,\xi_+]$, engaging in minimal trading only at the boundaries $\xi_\pm$ to maintain the statistic within the interval. This strategy is embodied by a diffusion with two reflecting boundaries and is referred to as a control limit policy (\cite{taksar1988diffusion}). The parameterization of the involved equations and the methods to solve them are customized for each specific setup. However, the analytical methods are highly technical, making them suitable for an expert audience. Moreover, the resulting optimal trading policies are reflected diffusions (cf.~\cite{tanaka}), making them impractical for implementation as they entail infinite activity.

The aim of this paper is to address these challenges by introducing asymptotic methods, both for approximating the value functions in various problems and designing asymptotically optimal trading policies that can be executed in real trading environment. Asymptotic methods have proved instrumental in
understanding sensitivity of trading boundaries and welfare to different market and risk-aversion parameters (cf.~ the ingenious insights of \cite{constantinides} based on numerical computation, with the financial interpretation of the asymptotic expansions derived in \cite{kallsenmad, gerhold2012asymptotics,gerhold.al.11}.

First, we employ finite activity strategies to approximate control limit policies. In contrast to reflected diffusions, these strategies can be implemented in a real trading environment. They involve trades of fixed size aimed at re-aligning a portfolio statistic with its frictionless target, such as adjusting the proportion of wealth $\pi_t$ in the risky asset when reaching $\pi_\pm$. Acknowledging the insightful contribution of \cite[Appendix B]{MR1284980}, which highlights that no-trade regions typically scale with size proportional to $\varepsilon^{1/3}$, discrete trades in this context must be of the same or higher order as the bid-ask spread. The findings of this paper offers an alternative answer to the question of \cite{rogers} regarding the $\varepsilon^{2/3}$ effect of proportional transaction costs on welfare. Approximating the control limit policy with trades of fixed size reveals that transaction costs, at the leading order, are a product of the bid-ask spread ($\varepsilon$), trade size ($\varepsilon^{\alpha/3}$ for $\alpha\geq 1$), and trading frequency (proportional to $\varepsilon^{-\frac{1}{3}(\alpha+1)}$), resulting in the same cost of order $\varepsilon^{2/3}$, irrespective of the smallness of trades (dictated by the parameter $\alpha$) or the objective under consideration.

Second, we put forward a general approximation method for pertinent free boundary problems, which often\footnote{For the Log Contract we use another parameterisation.} can be expressed as
\begin{align}\label{eq: ODE2}
& F(g,g',g'')=0,\\\label{eq: shadow constraints}
g(\pi_-)&=1,\quad g(\pi_+)=1-\varepsilon,\\\label{eq: smooth pasting}
g'(\pi_-)&=0, \quad g'(\pi_+)=0.
\end{align} 
Here, an unknown scalar $g=g(\pi)$ representing the proportion of wealth $\pi$ in the risky asset is involved. The function $g$ must satisfy a second-order ODE \eqref{eq: ODE2} between the buy and sell boundaries $\pi_-$ and $\pi_+$, and boundary conditions \eqref{eq: shadow constraints} and \eqref{eq: smooth pasting}. While the latter are universal, the ODE \eqref{eq: ODE2} identifies the specific objective under consideration. The approximation method relies on a polynomial expansion of the scalar function $g$, and an asymptotic Ansatz for the trading boundaries $\pi_\pm$.

The parameter $\varepsilon$ in the condition \eqref{eq: shadow constraints} represents the relative bid-ask spread of the risky asset $S$, where one purchases the asset at time $t$ at the ask price $S_t$ and sells it at the slightly lower price $(1-\varepsilon)S_t$ (bid). This boundary condition's specific form lends itself to the financial interpretation of a frictionless price $\widetilde S$ with a functional form expressed as $\widetilde S_t:=g(\pi_t)S_t$. This serves as a potential shadow price, that is fictitious risky asset evolving within the bid-ask spread $S_t(1-\varepsilon)\leq \widetilde S_t\leq S_t$ and which aligns precisely with the bid or ask of $S$ whenever the optimal strategy for $\widetilde S$ involves executing trades. This asset is indeed a shadow price for power- and log-utility, as noted in previous studies \cite{gerhold2013dual,gerhold.al.11,guasonimuhlekarbe}. However, for certain local criteria, it might only serve as an asymptotic shadow price (or none at all), as discussed in \cite{kallsenmad,em_1_2024}. Notwithstanding this issue, our method is applicable for equations both in the original markets, as in potential shadow markets.
A distinctive aspect of this paper lies in its examination of Neuberger's log contract \cite{neuberger} and its dynamic hedge under proportional transaction costs, supplementing the well-explored objectives with known solutions in the literature. Due to the straightforward form of its free boundary problem, our analysis may serve as a valuable introductory example for graduate students delving into the realm of optimization under proportional transaction costs.

\subsection*{Program of Paper}
The paper's structure is organized as follows: Section \ref{sec: the model} introduces the market model, encompassing a risky Black-Scholes asset with transaction costs, along with the presentation of relevant admissible strategies for this study. Section \ref{sec: Results} comprises three distinct parts:

Section \ref{sec: transaction costs} investigates three strategies and evaluates their performance in terms of trading frequency and long-run cost. This includes the examination of minimal trades (control limit policies, see Section \ref{sec: minimal trades}), maximal trades (discrete trades directly to the target, see Section \ref{sec: maximal trades}), and small trades (sufficiently small trades to closely align with control limit policies, see Section \ref{sec: micro trades}).
Section \ref{sec: tracking} is dedicated to optimization problems related to tracking specific targets. It includes the computation of asymptotic approximations for the free boundary problem for Leveraged ETFs (see Section \ref{sec: LETFS}) and optimally hedging a log-Contract (see Section \ref{sec: log contract}).
Section \ref{sec: risk} examines classical objectives with a long-run horizon, featuring the risk-averse investor (Section \ref{sec: power utility}), the log-utility investor (Section \ref{sec: log utility}), and the risk-neutral investor (Section \ref{sec: risk neutral}).
The interpretation and discussion of these results are presented in Section \ref{sec: Discussion}.

\section{Materials and Methods}\label{sec: the model}

In our financial market model, there are two tradable assets: a risk-free asset with continuous compounding 
at a rate $r \geq 0$ and a risky asset $S$ with ask price $S_t$, governed by the dynamics
\begin{equation}\label{eq:evolution}
\frac{dS_t}{S_t} = (\mu + r)dt + \sigma dB_t, \quad S_0, \sigma > 0, \quad \mu \geq 0,
\end{equation}
where $B$ denotes a standard Brownian motion. We introduce a relative bid-ask spread $\varepsilon > 0$, ensuring that the risky asset is sold at $(1 - \varepsilon)S_t$.

In markets with frictions, trading strategies need to be modified so to ensure solvency.
In the presence of transaction costs, solvency constraints dictate that only finite-variation trading strategies are viable.\footnote{Discrete trades occur in the context of fixed transaction costs \cite{alta}, where only a fixed number of trades maintains solvency. Finite trading speed -that is absolute continuity of the number of shares of the risky asset- occurs in market impact models, cf. \cite{guasoniweber} and the references therein. See also \cite{annals}.} Let $X_t$ be the safe position and $Y_t$ the risky position. Furthermore, denote by $(\varphi_t^\uparrow)_{t\ge 0}$ and $(\varphi_t^\downarrow)_{t\ge 0}$ the cumulative number of shares bought and sold, respectively. The total wealth, assessed at the ask, is given by $w_t=X_t+Y_t=X_t+\varphi_t S_t$, where $\varphi_t=\varphi_t^{\uparrow}-\varphi_t^{\downarrow}$.

In this paper, we study two types of trading strategies: For continuous $\varphi_t$, the self-financing condition reads
\begin{align}\label{eq: sf1}
X_t &= X_0+\int_0^t rX_u du-\int_0^t S_u d\varphi^\uparrow_u+(1-\varepsilon)\int_0^t S_u d\varphi^\downarrow_u,\\\nonumber
Y_t &=Y_0+ \int_0^t S_u d\varphi^\uparrow_u- \int_0^t S_u d \varphi^\downarrow_u+\int_0^t \varphi_u dS_u.
\end{align}
If $\varphi_t$ is a finite activity jumps process of pure jump type, such that
\[
\varphi_t^\uparrow=\sum_{s\leq t}\Delta \varphi^\uparrow_s,\quad \varphi_t^\downarrow=\sum_{s\leq t} \Delta\varphi^\downarrow_s,
\]
then the self-financing condition reads
\begin{align}\label{eq: sf2}
X_t &= X_0+\int_0^t rX_u du-\sum_{u\leq t} S_u \Delta \varphi^{\uparrow }_t+(1-\varepsilon)\sum_{u\leq t} S_u \Delta \varphi^\downarrow_u,\\\nonumber
Y_t &=Y_0+ \sum_{u\leq t} S_u \Delta \varphi^\uparrow_u-\sum_{u\leq t} S_u \Delta \varphi^\downarrow_u+\int_0^t \varphi_{u_-} dS_u.
\end{align}
The following defines admissible strategies:

\begin{definition}\label{def: admiss}
The trading policy $(x=X_0+\varphi_0 Y_0, \varphi_t=\varphi_t^\uparrow-\varphi_t^\downarrow)$ is \emph{admissible} if
\begin{enumerate}
\item $(\varphi_t^\uparrow)_{t\ge 0}$ and $(\varphi_t^\downarrow)_{t\ge 0}$ are right-continuous, increasing, and non-anticipating processes.
\item It is self-financing, satisfying \eqref{eq: sf1} resp. \eqref{eq: sf2}.
\end{enumerate}
The collection of admissible trading policies is denoted by $\Phi$.
\end{definition}

Depending on the specific optimization criterion, some extra integrability or solvency conditions may apply. (Since we focus on asymptotic methods for solving a range of free boundary problems, we refrain from formulating them explicitly.)

The next statement elaborates the dynamics of wealth, its fraction in the risky asset, and the ratio of risky and safe investment. 
\begin{lemma}\label{le: rewriting obj fun}
For any admissible, continuous trading strategy $\varphi$, we have
\begin{align}\label{eq zeta diff}
\frac{d\zeta_t}{\zeta_t}&=\mu dt +\sigma dB_t+(1+\zeta_t)\frac{d\varphi_t^\uparrow}{\varphi_t} -(1+(1-\varepsilon)\zeta_t)\frac{d\varphi_t^\downarrow}{\varphi_t},\\\label{eq w diff}
\frac{dw_t}{w_t}&=r dt+\pi_t(\mu dt+\sigma dB_t-\varepsilon \frac{d\varphi_t^\downarrow}{\varphi_t}),\\\label{eq pi diff}
\frac{d\pi_t}{\pi_t}&=(1-\pi_t)(\mu dt+\sigma dB_t)-\pi_t(1-\pi_t)\sigma^2 dt+\frac{d\varphi_t^\uparrow}{\varphi_t}-(1-\varepsilon\pi_t)\frac{d\varphi_t^\downarrow}{\varphi_t}.
\end{align}
If $\varphi_t$ is of pure jump type (cf.~\eqref{eq: sf2}, the same statistics satisfy similar equations, perhaps most rigorously stated in integral form,
\begin{align}\label{eq zeta diff1}
\zeta_t&=\zeta_0+\int_0^t\zeta_{s_-}\left(\mu ds +\sigma dB_s\right)+\sum_{s\leq t}\zeta_{s_-}(1+\zeta_{s_-})\frac{\Delta \varphi_s^\uparrow}{\varphi_{s_-}} \\\nonumber\qquad\qquad &-\sum_{s\leq t}\zeta_{s_-}(1+(1-\varepsilon)\zeta_{s_-})\frac{\Delta\varphi_{s}^\downarrow}{\varphi_{s_-}},\\\label{eq w diff1}
w_t&=w_0+r\int_0^t w_{s_-} ds+\int_0^t w_{s_-}\pi_{s_-}(\mu ds+\sigma dB_s)-\sum_{s\leq t}\varepsilon w_{s_-} \frac{\Delta \varphi_s^\downarrow}{\varphi_{s_-}},\\\label{eq pi diff1}
\pi_t&=\pi_0+\int_0^t\pi_{s_-}(1-\pi_{s_-})(\mu ds+\sigma dB_s)-\int_0^t\pi_{s_-}^2(1-\pi_{s_-})\sigma^2 ds\\\nonumber\qquad\qquad\qquad&+\sum_{s\leq t}\pi_{s_-}\frac{\Delta \varphi_s^\uparrow}{\varphi_{s_-}}-\sum_{s\leq t}\pi_{s_-}(1-\varepsilon\pi_{s_-})\frac{\Delta\varphi_s^\downarrow}{\varphi_{s_-}}.
\end{align}
\end{lemma}
\begin{proof}
The dynamics \eqref{eq zeta diff} are stated and proved in \cite[Lemma A.2]{gm2020}, which can be adapted for discrete trades, so to obtain \eqref{eq zeta diff1}.
\end{proof}

\section{Results}\label{sec: Results}

\subsection{Transaction Costs Asymptotics}\label{sec: transaction costs}

\begin{figure}
\centering
\includegraphics[width=1\textwidth]{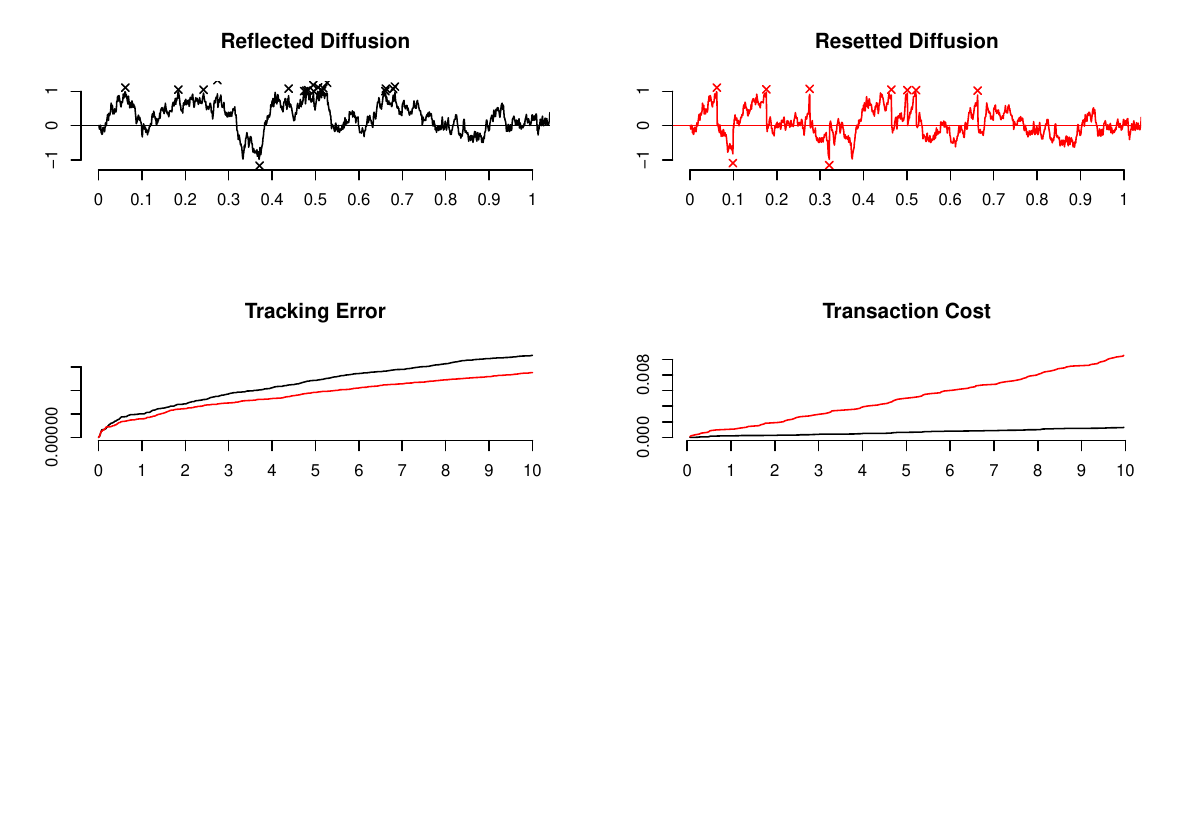}
\caption{The upper panel plots illustrate reflected (black line) and resetted (red line) centred Brownian motion with a standard deviation of 10\%. Reflected diffusions execute minimal trades - here approximated with small trades - while maximal trades are characteristic of resetted diffusions (as depicted in the second plot). Crosses on the graphs mark trading dates and positions. Notably, over a ten-year period, annualized tracking error tends to be smaller for resetted diffusions, whereas transaction costs are generally lower for minimal trades (lower panel).}
\label{fig: simulating trades}
\end{figure}

Reflected diffusions are continuous processes and of finite variation by construction, therefore strategies of this kind are admissible in the presence of proportional transaction costs. While optimal across a range of optimization criteria in theory, 
they are practically unrealized due to actual trades occurring discretely in time. 
In this chapter, we approximate continuous trading policies with such that trade fixed amounts of the risky asset, thus introduce jumps. The primary challenge lies in quantifying the minimal 
trade size required to achieve nearly optimal performance. While larger trades align better with a target, they also incur greater costs (cf.~Figure \ref{fig: simulating trades}).

In this section, we investigate transaction costs and trading frequency associated with three strategies aimed at maintaining a portfolio statistic near a target. These strategies include:

\begin{enumerate}
\item Minimal trades: These strategies operate continuously, engaging in minimal trading at the boundaries of specific no-trade regions (cf.~the first plot of Figure \ref{fig: simulating trades} (black lines)),
\item \label{item: 2} Maximal trades: These strategies reset a portfolio statistic to its target value once it reaches these boundaries (cf.~the first second plot in Figure \ref{fig: simulating trades} (red lines)),
\item Small trades: These involve smaller trades than \ref{item: 2} made in the direction of the target value.
\end{enumerate}

The general framework involves a geometric Brownian motion $\xi$ with a drift parameter $M$ and diffusion parameter $\Sigma$ evolving in an interval $[\xi_-,\xi_+]$, reflected at its boundaries $\xi_\pm$. The financial interpretation of $\xi$ varies among examples, representing either the risky-safe ratio $\zeta_t$ or the risky position $Y_t$. Similarly, the processes that confine the GBM within $[\xi_-,\xi_+]$ vary, influencing the definition of transaction costs.

Each subsection begins by outlining this general setup of a reflected or resetted geometric Brownian motion. We introduce two types of transaction costs: average transaction costs (ATC), measuring the negative contribution to infinitesimal returns due to non-zero bid-ask spread, and transaction costs (TRC), which quantify the actual long-term cost per time unit.
ATC arises for Levered ETFS in Section \ref{sec: LETFS} and for the risk-neutral investor in Section \ref{sec: risk neutral} as a penalty in the objective\footnote{Actually, also for the log-utility, the same penalty appears: rewrite the finite-horizon version of objective \eqref{eq: log criterion}, using It\^o's formula.}, whereas TRC provides the penalty for the log-contract in Section \ref{sec: log contract}.

\subsubsection{Minimal Trades}\label{sec: minimal trades}
The dynamics of a reflected GBM are
\begin{equation}\label{eq: GBM refl}
d\xi_t=\xi_t\left( M dt+\Sigma dW_t+dL_t-dU_t\right),
\end{equation}
where $L$ resp.~$U$ are local times, with the only points of increase (resp.~decrease) on $\xi_t=\xi_-$ (resp.~ $\xi_t=\xi_+$).

We shall need the stationary density of a reflected diffusion, cf.~\cite[Lemma B.5]{gm2020}:
\begin{lemma}\label{lem: existence controllable strategy}
Let $\xi_-<\xi_+$. Then there exists an admissible trading strategy ${\hat\varphi}$ such that the risky-safe ratio $\xi_t$ satisfies SDE \eqref{eq zeta diff}. Moreover, $(\xi_t,{\hat\varphi}_t^\uparrow,{\hat\varphi}_t^\downarrow)$ is a reflected diffusion on the interval $[\xi_-,\xi_+]$. In particular, if $\xi_->0$, then $\xi_t$ has stationary density equals
\begin{equation}\label{eq st de 1}
\nu(\xi):=\frac{\frac{2\mu}{\sigma^2}-1}{\xi_+^{\frac{2\mu}{\sigma^2}-1}-\xi_-^{\frac{2\mu}{\sigma^2}-1}}\xi^{\frac{2\mu}{\sigma^2}-2}, \quad \xi\in[\xi_-,\xi_+],
\end{equation}
and if $\xi_+<-1$, then
\begin{equation}\label{eq st de 2}
\nu(\xi):=\frac{\frac{2\mu}{\sigma^2}-1}{\vert\xi_-\vert^{\frac{2\mu}{\sigma^2}-1}-\vert\xi_+\vert^{\frac{2\mu}{\sigma^2}-1}}\vert\xi\vert^{\frac{2\mu}{\sigma^2}-2}, \quad \xi\in[\xi_-,\xi_+].
\end{equation}
\end{lemma}
The trading strategy in Lemma \ref{lem: existence controllable strategy} is called ``control limit policy", cf.~ \cite{taksar1988diffusion,em_1_2024}.

We start with the following:
\begin{lemma}\label{lem: GBM LU}
\begin{align}\label{eq: L}
L_\infty&:=\lim_{T\rightarrow\infty}\frac{1}{T}\int_0^T dL_t=\lim_{T\rightarrow\infty}\frac{1}{T}L_t=\frac{\Sigma^2}{2}\frac{2 M/\Sigma^2-1}{\left((\xi_+/\xi_-)^{2M/\Sigma^2-1}-1\right)},\\\label{eq: U}
U_\infty&:=\lim_{T\rightarrow\infty} \frac{1}{T}U_t=\frac{\Sigma^2}{2}\frac{1-2 M/\Sigma^2}{\left((\xi_+/\xi_-)^{1-2M/\Sigma^2}-1\right)}.
\end{align}
\end{lemma}
\begin{proof}
(Cf.~ \cite[Lemma C.2]{gerhold.al.11} which is slightly different version, with a different proof.) For convenience, we perform a logarithm transformation: The process $\eta_t:=\log(\xi_t)$ is a reflected Brownian motion on $[\eta_-,\eta_+]$, where $\eta_\pm=\log(\xi_\pm)$, with dynamics
\[
d\eta_t= (M-\Sigma^2/2) dt+\Sigma dW_t+dL_t-dU_t.
\]
To extract $L_\infty$, apply It\^o's formula to the function $f_-(\eta)=\frac{1}{2}\frac{(\eta_+-\eta)^2}{\eta_+-\eta_-}$. As $f'(\eta_+)=0$, and $f'(\eta_-)=-1$, we get
\[
f(\eta_T)=f(\eta_0)-L_T+(M-\Sigma^2/2) \int_0^T f'(\eta_t)dt+\Sigma \int_0^T f'(\eta)dW_t+\frac{\Sigma^2}{2}\frac{1}{\eta_+-\eta_-}T.
\]
Dividing by $T$ and letting it tend to infinity, we have
\[
\lim_{T\rightarrow \infty}\frac{L_T}{T}=\lim_{T\rightarrow\infty}\frac{M-\Sigma^2/2}{T}\int_0^T f'(\eta_t)dt+\frac{\Sigma^2}{2(\eta_+-\eta_-)}.
\]
The limit on the right side can be computed by the Ergodic Theorem \cite[Chapter II]{borodin} such that
\[
\lim_{T\rightarrow\infty}\frac{1}{T}\int_0^T f'(\eta_t)dt=\int_{\eta_-}^{\eta_+}f'(u)g(u)du,
\]
where the stationary density $g$ of $\eta$ follows from Lemma \ref{lem: existence controllable strategy}
\[
g(\eta)=\frac{\frac{2M}{\Sigma^2}-1}{e^{(\frac{2M}{\Sigma^2}-1)\eta_+}-e^{(\frac{2M}{\Sigma^2}-1)\eta_-}} e^{(\frac{2M}{\Sigma^2}-1)\eta},\quad \eta\in[\eta_-,\eta_+].
\]
In total, 
\[
\lim_{T\rightarrow \infty}\frac{L_T}{T}=\frac{\Sigma^2}{2}\frac{2 M/\Sigma^2-1}{\left(e^{(\eta_+-\eta_-)(2M/\Sigma^2-1)}-1\right)},
\]
from which \eqref{eq: L} follows. A similar argument produces \eqref{eq: U}.

\end{proof}

We consider now two important examples, which both are encountered in this paper. The first one trades, so to keep the risky position approximately constant, a strategy that is
required in the optimal replication of a Log contract:
\begin{proposition}\label{prop: constant Y strategy}
For $y_\star>0$, consider the strategy that trades minimally, by only buying at $y_-<y_\star$ and selling at $y_+>y_\star$ as much as necessary to keep the
risky-position $Y_t$ in the interval $[y_-,y_+]$. If the trading boundaries are of the form\footnote{For cosmetic reasons, which shall be clear later, we subtract the second order term from each boundary.
However, since $B_\pm$ may be any real number, our choice does not lead to a loss of generality.}
\begin{equation}\label{eq: ans boundaries Y} 
y_\pm=y_\star \pm A_\pm \varepsilon^{1/3}- B_\pm\varepsilon^{2/3}+O(1),
\end{equation}
with $A_\pm>0$, and $B_\pm\in\mathbb R$, then the long-run, average transaction costs are
\begin{equation}\label{eq: ATC GBM exactd3}
\text{TRC}:=\lim_{T\rightarrow\infty}\frac{\varepsilon}{T}\int_0^T S_t d\varphi^\downarrow_t=\varepsilon y_+\times \frac{\sigma^2}{2}\frac{1-2 (\mu+r)/\sigma^2}{\left((y_+/y_-)^{1-2(\mu+r)/\sigma^2}-1\right)}.
\end{equation}
For sufficiently small $\varepsilon$, we have the following asymptotic formula:
\begin{equation}\label{eq: ATC GBM asy1}
\text{TRC}=\frac{\sigma^2 y_\star^2}{2(A_++A_-)}\varepsilon^{2/3}+\frac{\sigma^2}{2} \left(\frac{\mu+r}{\sigma^2}+\frac{A_+^2-A_-^2+y_\star (B_+-B_-)}{(A_++A_-)^2}\right)\varepsilon+O(\varepsilon^{4/3}).
\end{equation}
\end{proposition}
\begin{proof}
Since the risky position $Y_t=\varphi_t S_t$, we have $\text{TRC}=\lim_{T\rightarrow\infty}\frac{\varepsilon}{T}\int_0^T Y_t \frac{d\varphi^\downarrow_t}{\varphi_t}$. It follows that $Y_t$ is a reflected Geometric Brownian motion such as $Y_t=\xi_t$ in \eqref{eq: GBM refl} with $M=\mu+r$, $\Sigma=\sigma$
and $d L_t=\frac{d\varphi^\uparrow_t}{\varphi_t}$, $dU_t=\frac{d\varphi^\downarrow_t}{\varphi_t}$. Thus formula \ref{eq: ATC GBM exactd3} is an application of Lemma \ref{lem: GBM LU}. Plugging the asymptotics of the trading boundaries \eqref{eq: ans boundaries Y} into \eqref{eq: ATC GBM exactd3}, we produce a formal power series in $\delta=\varepsilon^{1/3}$, to get
\eqref{eq: ATC GBM asy1}.
\end{proof}

The next strategy trades, so to keep the proportion of wealth approximately constant: 
\begin{proposition}\label{prop: const prop minimal trades}
For $\Lambda>0$, consider the strategy that trades minimally, by only buying at $\pi_-<\Lambda$ and selling at $\pi_+>\Lambda$ as much as necessary to keep the
proportion of wealth in the risky asset $\pi_t$ in the interval $[\pi_-,\pi_+]$. If the trading boundaries are of the form
\begin{equation}\label{eq: ans boundaries Yz} 
\pi_\pm=\Lambda \pm A_\pm \varepsilon^{1/3}- B_\pm\varepsilon^{2/3}+O(1),
\end{equation}
with $A_\pm>0$, and $B_\pm\in\mathbb R$, then the long-run, average transaction costs are
\begin{equation}\label{eq: ATC GBM exact}
\text{ATC}:=\lim_{T\rightarrow\infty}\frac{\varepsilon}{T}\int_0^T \pi_t\frac{d\varphi^\downarrow_t}{\varphi_t}=\varepsilon \frac{\pi_+(1-\pi_+)}{1-\varepsilon \pi_+}\times \frac{\sigma^2}{2}\frac{1-2 (\mu+r)/\sigma^2}{\left((\frac{\pi_+(1-\pi_-)}{\pi_-(1-\pi_+)})^{1-2(\mu+r)/\sigma^2}-1\right)}.
\end{equation}
For sufficiently small $\varepsilon$, we have the following asymptotic formula:
\begin{align}\label{eq: ATC GBM asy}
\text{ATC}&=\frac{\sigma^2 \Lambda^2(\Lambda-1)^2}{2(A_++A_-)}\varepsilon^{2/3}+\frac{\sigma^2 (\Lambda-1)\Lambda}{2 (A_++A_-)^2} \\\nonumber&\times \left(-\frac{\mu}{\sigma^2}(A_+-A_-)^2-A_-^2 (\Lambda -1)+2 A_-A_+ \Lambda +A_+^2 (3
\Lambda -1)-(\Lambda -1) \Lambda (B_--B_+)\right)\varepsilon\\\nonumber&+O(\varepsilon^{4/3}).
\end{align}
\end{proposition}
\begin{proof}
We change to the risky-safe ratio $\zeta_t=\frac{\pi_t}{1-\pi_t}$, such that the new trading boundaries have asymptotics
\[
\zeta_\pm=\frac{\Lambda}{1-\Lambda}\pm \frac{A_\pm}{(\Lambda-1)^2}\varepsilon^{1/3}-\frac{A_\pm^2+B_\pm(\Lambda-1)}{(\Lambda-1)^3}\varepsilon^{2/3}+O(\varepsilon).
\]
The risky-safe ratio satisfies the dynamics \eqref{eq zeta diff} and thus is a Geometric Brownian motion of the form \eqref{eq: GBM refl}, where $M=\mu$ and $\Sigma=\sigma$, while
\[
dL_t=(1+\zeta_t)\frac{d\varphi_t^\uparrow}{\varphi_t},\quad d U_t=(1+(1-\varepsilon)\zeta_t)\frac{d\varphi_t^\downarrow}{\varphi_t}.
\]
The rest of the proof is similar to that of Proposition \ref{prop: constant Y strategy}.
\end{proof}
\begin{remark}\label{rem: second order same sign impact no}
The second order coefficients $B_\pm$ of the trading boundaries \eqref{eq: ans boundaries Yz} 
only enter the third order term of the average trading costs \eqref{eq: ATC GBM asy}, and their impact depends on the difference $B_+-B_-$. For a range of criteria, including the ones in the present paper, we have $B_+=B_-$, whence these coefficients do not impact the average trading costs at third order. 
\end{remark}

\subsubsection{Maximal Trades}\label{sec: maximal trades}
In this section, we consider the process $\xi_t$ on $[\xi_-,\xi_+]$ that satisfies
\begin{equation}\label{eq: fixed jumps zeta}
\xi_t=\xi_0+\int_0^t \xi_{s_-}\left( Mds+\Sigma dW_s\right)+\sum_{s\leq t}\Delta L_s-\sum_{s\leq t}\Delta U _s
\end{equation}
with fixed jump sizes, such that we reset $\xi_t$ to some point $\xi_\star\in (\xi_-,\xi_+)$, whenever it hits $\xi_\pm$. We compute the average, long-run jump sizes, and then compute transaction costs for constant proportion trading strategies similar as in Section \ref{sec: minimal trades}. By taking logarithms we get $\eta_t:=\log(\xi_t)$, which is a resetted Brownian motion on $[\eta_-,\eta_+]$, where $\eta_\pm=\log(\xi_\pm)$, satisfying the dynamics
\begin{equation}\label{eq: eta with jumps}
\eta_t=\eta_0+ (M-\Sigma^2/2)t+\Sigma W_t+\sum_{s\leq t}\Delta l_s-\sum_{s\leq t}\Delta u_s.
\end{equation}
Whenever $\eta_t$ hits $\eta_-$ (resp. $\eta_+$), it is resetted to $\eta_\star=\log(\xi_\star)$. Thus we have $\Delta l=\eta_\star-\eta_-=\log(\xi_\star/\xi_-)$, and $\Delta u=\eta_+-\eta_\star=\log(\xi_+/\xi_\star)$. Define the differential operator
\[
\mathcal A f(x)=\frac{\Sigma^2}{2}f''(x)+(M-\Sigma^2/2)f'(x).
\]
By construction, $\eta$ is a Markov process on $[\eta_-,\eta_+]$.\footnote{In fact, $\eta_t$ is the unique strong solution to \eqref{eq: eta with jumps}.} By It\^o's formula, for any $f\in C^2([\eta_-,\eta_+])$,
\begin{align}\label{eq: Ito max}
f(\eta_t)&=f(\eta_0)+\int_0^t \sigma f'(\eta_s)dW_s+\int_0^t \mathcal A f(\eta_s)ds\\\nonumber&+(f(\eta_\star)-f(\eta_-))\vert \{s\leq t: \Delta L_s>0\}-(f(\eta_+)-f(\eta_\star))\vert \{s\leq t: \Delta U_s>0\}.
\end{align}
Thus, $\eta$ solves the martingale problem for $(\mathcal A,\mathcal D(\mathcal A))$, where $\mathcal D(\mathcal A)$ constitutes those functions $f\in C^b([\eta_-,\eta_+])$ for which $f(\eta_-)=f(\eta_\star)=f(\eta_+)$ and thus $f(\eta_t)-\int_0^t \mathcal A f(\eta_s)ds$ is a martingale.
\begin{lemma}\label{lem: stationary distribution}
The process $\eta$ is ergodic, with stationary density
\begin{equation}\label{eq: stationary density maximal trades}
\varphi(\eta)=\begin{cases}a_1+b_1 e^{\alpha \eta}\quad\text{for}\quad \eta_-\leq \eta\leq \eta_\star,\\
a_2+b_2 e^{\alpha \eta}\quad\text{for}\quad \eta_\star\leq \eta\leq \eta_+,\end{cases}
\end{equation}
where $\alpha=2M/\Sigma^2-1$ and
\begin{align*}
a_1&=-b_1 e^{\alpha\eta_- },\quad
a_2=-b_2 e^{\alpha\eta_+ },\\
b_1&=\frac{e^{\alpha \eta_\star}-e^{\alpha \eta_+}}{(\eta_--\eta_\star) e^{\alpha(\eta_-+\eta_\star)}+(\eta_+-\eta_-) e^{\alpha(\eta_++\eta_-)}+(\eta_{\star}-\eta_+) e^{\alpha(\eta_\star+\eta_+)}
},\\
b_2&=\frac{e^{\alpha \eta_\star}-e^{\alpha \eta_-}}{(\eta_--\eta_\star) e^{\alpha(\eta_-+\eta_\star)}+(\eta_+-\eta_-) e^{\alpha(\eta_++\eta_-)}+(\eta_{\star}-\eta_+) e^{\alpha(\eta_\star+\eta_+)}
}.
\end{align*}
The ergodic theorem applies, such that for any bounded continuous\footnote{This results holds for all essentially bounded, Borel-measurable functions.} function $f$
\begin{equation}\label{eq: limit problem}
\lim_{T\rightarrow\infty} \frac{1}{T}\int_0^T f(\eta_s)ds=\int_{\eta_-}^{\eta_+} f(u)\varphi(u)du.
\end{equation}
\end{lemma}
\begin{proof}
To derive the stationary density $\varphi$, we solve the adjoint equation
\[
\mathcal A^\star \varphi=\frac{\Sigma^2}{2}\varphi''(\eta)-(M-\Sigma^2/2) \varphi'(\eta)=0
\] 
for $\varphi^1$ on $(\eta_-,\eta_\star)$ and for $\varphi^2$ on $(\eta_\star,\eta_+)$, where $\varphi^i(\eta)=a_i+b_i e^{\alpha \eta}$ ($i=1,2$).
The constants $a_i$ and $b_i$ are determined by insisting on continuity ($\varphi^1(\eta_\star)=\varphi^2(\eta_\star)$), vanishing density at the boundaries ($\varphi^1(\eta_-)=\varphi^2(\eta_+)=0$) and mass one ($\int_{\eta_-}^{\eta_+} \varphi(u)du=1$). 

To prove ergodicity, we assume for brevity of exposition that $M=\Sigma^2/2$ and that $\eta_\star=\frac{\eta_-+\eta_+}{2}$ (the general case can be proved similarly, but is more tedious due to lengthy formulas). Furthermore, by shifting the process, we can assume without loss of generality that $\eta_\star=0$. In this case the stationary density takes the specific form
\begin{equation}\label{st1}
\varphi(\eta)=\begin{cases}\frac{2(\eta-\eta_-)}{(\eta_\star-\eta_-)(\eta_+-\eta_-)},\quad \eta_-\leq\eta\leq\eta_\star,\\
\frac{2(\eta_+-\eta)}{(\eta_+-\eta_\star)(\eta_+-\eta_-)},\quad \eta_\star\leq \eta\leq\eta_+.\end{cases}
\end{equation}
For computing the limit on the left side of \eqref{eq: limit problem}, the starting point $\eta_0=\eta$ is irrelevant,
as the diffusion is resetted every time it hits the boundary (the integral up to the first hitting time gives a vanishing contribution to the overall limit), thus starting from $\eta_0=\eta_\star=0$. Therefore we may without loss of generality start the process $\eta_t$ also initially at $\eta_0=0$. Denote the consecutive exit times
\[
\tau_1:=\inf\{t>0\mid X_t\notin (\eta_-,\eta_+)\},\quad \tau_k:=\inf\{t>\tau_{k-1}\mid X_t\notin (\eta_-,\eta_+)\},\quad k\geq 2.
\]
Since $\eta_t$ is Markovian, the inter-exit times $T_k:=\tau_k-\tau_{k-1}$ ($k\geq 2$) are independent and identically distributed $\sim \tau_1$ with first moment
\begin{equation}\label{eq: moment1}
\mathbb E[\tau_1]=\frac{-\eta_-\eta_+}{\Sigma^2}.
\end{equation}
Thus by the strong law of large numbers (SLLN), almost surely
\begin{equation}\label{eq: average interexit times}
\lim_{N\rightarrow \infty}\frac{1}{N}\sum_{k=1}^\infty T_k=\frac{-\eta_-\eta_+}{\Sigma^2}.
\end{equation}
Similarly, the random variables $Y_k:=\int_{\tau_{k-1}}^{\tau_k} f(\eta_s)ds$ are independent, identically distributed with mean $\mathbb E[Y_1]=\mathbb E[\int_0^{\tau_1} f(\eta_s)ds]$
and thus by the SLLN, almost surely
\begin{equation}\label{eq: moment 1}
\lim_{N\rightarrow \infty}\frac{1}{N}\sum_{k=1}^N Y_k=\mathbb E[\int_0^{\tau_1} f(\eta_s)ds].
\end{equation}
The expectation on the right hand of \eqref{eq: moment 1} is computed in the following. Let $w$ be the unique solution of the Dirichlet problem
\begin{equation}\label{eq: dirichlet}
\mathcal A w+f=0,\quad w(\eta_-)=w(\eta_+)=0.
\end{equation}
By It\^o's formula for $0\leq t< \tau_1$,
\[
w(\eta_t)=w(\eta_0)+\int_0^t \mathcal A w(\eta_s)ds+\int_0^t \sigma w'(\eta_s)dW_s.
\]
Since $w(\eta_{\tau_1}-)=0$, $\eta_0=0$, and \eqref{eq: dirichlet} we get by optimal stopping
\[
w(0)=\mathbb E[\int_0^{\tau_1} f(\eta_s)ds].
\]
The explicit solution of \eqref{eq: dirichlet} is given by
\[
w(\eta)=\frac{2}{\sigma^2}\left(\frac{\eta-\eta_-}{\eta_+-\eta_-}\int_{\eta_-}^{\eta_+}(\eta_+-\xi)f(\xi)d\xi-\int_{\eta_-}^\eta (\eta -\xi)f(\xi)d\xi\right).
\]
Comparing with the stationary density \eqref{st1} reveals that
\[
w(0)=\frac{\vert\eta_-\eta_+\vert}{\Sigma^2}\mathbb E[\int_0^{\tau_1} f(\eta_s)ds]=\frac{\vert\eta_-\eta_+\vert}{\Sigma^2}\int_{\eta_-}^{\eta_+}f(\xi)\varphi(\xi)d\xi
\]
and thus by \eqref{eq: moment 1}, 
\[
\frac{\mathbb E[\int_0^{\tau_1} f(\eta_s)ds]}{\mathbb E[\tau_1]}=\int_{\eta_-}^{\eta_+}f(\xi)\varphi(\xi)d\xi.
\]
A combination with \eqref{eq: average interexit times} and \eqref{eq: moment 1} yields thus
\begin{equation}\label{eq: tau limit}
\lim_{N\rightarrow \infty}\frac{1}{\tau_N}\int_0^{\tau_N} f(\eta_s)ds=\frac{\lim_{N\rightarrow \infty}\frac{\sum_{k=1}^N Y_k}{N}}{\lim_{N\rightarrow \infty}\frac{\sum_{k=1}^N T_k}{N}}=\frac{\mathbb E[\int_0^{\tau_1} f(\eta_s)ds]}{\mathbb E[\tau_1]}=\int_{\eta_-}^{\eta_+}f(\xi)\varphi(\xi)d\xi.
\end{equation}
Note that this limit involves annualizing over random horizons the $\tau_N$. It remains to show that this limit is the same for annualizing with deterministic time $T$. For each $t\geq 1$, let $\nu(t)$ be the first passage time of $t$, for the renewal process $(\tau_n)_{n\geq 1}$, where
$\tau_n:=\sum_{k=1}^n T_k$. That is,
\[
\nu(t):=\min\{n\mid \tau_n > t\}.
\]
A combination of \cite[Theorem 4.1 and Theorem 3.4.2]{gut} yields that
almost surely,
\begin{equation}\label{eq: gut1}
\lim_{n\rightarrow\infty}\frac{\nu(t)}{t}=\frac{1}{\mathbb E[\tau_1]},
\end{equation}
and 
\begin{equation}\label{eq: gut2}
\lim_{n\rightarrow\infty}\frac{\tau_{\nu(n)}}{n}=1.
\end{equation}
Therefore,
\begin{align*}
&\lim_{n\rightarrow\infty}\frac{1}{n}\int_0^n f(\eta_s)ds=\lim_{n\rightarrow\infty}\frac{\tau_{\nu(n)}}{n}\left(\frac{1}{\tau_{\nu(n)}}\int_0^{\tau_{\nu(n)}}f(\eta_s)ds-\frac{1}{\tau_{\nu(n)}}\int_n^{\tau_{\nu(n)}}f(\eta_s)ds\right)\\
&=\lim_{n\rightarrow\infty}\frac{1}{\tau_n}\int_0^{\tau_n} f(\eta_s)ds,
\end{align*}
where because of \eqref{eq: gut1}, \eqref{eq: gut2} and \eqref{eq: tau limit} the first term converges to the desired limit, whereas, the second term, bounded by
\[
\|f\|_{\infty}\frac{\tau_{\nu(n)}-n}{\tau_{\nu(n)}}
\]
converges due to \eqref{eq: gut2} to zero, as $n\rightarrow\infty$. Recalling \eqref{eq: tau limit} yields the claim.
\end{proof}

In Figure \ref{fig: simulating density maximal trades}, the histogram of the (approximate) distribution of $\eta_t$ for $t=5$ years is given, along with
the exact stationary distribution provided by Lemma \ref{lem: stationary distribution}. The fit of this asymmetric distribution appears to be excellent, reflecting
that daily sampling is sufficient for financial data with reasonable parameter ranges, and that $\eta$ approaches stationarity for a finite time horizon of five years. 
\begin{figure}
\centering
\includegraphics[width=0.6\textwidth]{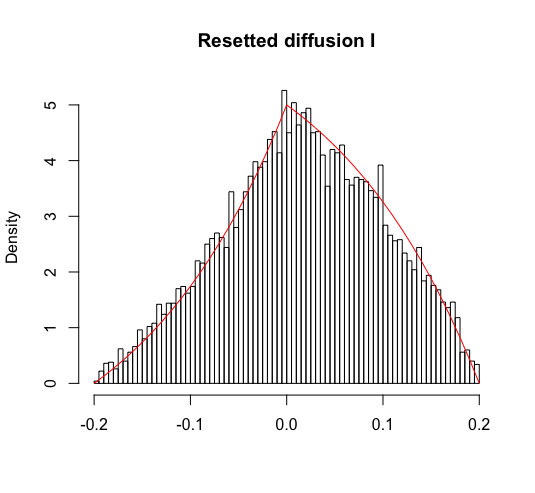}
\caption{For $M-\Sigma^2/2=8\%$ and $\Sigma=16\%$, and $\eta_-=-.2$, $\eta_\star=0$, and $\eta_+=+0.2$, we plot the stationary density of the resetted process
$\eta$. The histogram gives the value of $\eta_T$ for $T=5$ years, from ten thousand simulated path (daily frequency). }
\label{fig: simulating density maximal trades}
\end{figure}

We are able to count the number of downward jumps of $\eta$, by appealing to the ergodic theorem:
\begin{proposition}
\begin{equation}\label{eq: max jumps exact}
\lim_{T\rightarrow\infty}\frac{\vert\{s<T: \Delta U_s>0\}\vert}{T}=\frac{-(\eta_-+\eta_\star) \left(M-\frac{\Sigma ^2}{2}\right)+\frac{1}{2} \left(2 M-\Sigma ^2\right) \Theta+\Sigma ^2}{(\eta_+-\eta_-) (\eta_+-\eta_\star)},
\end{equation}
where
\begin{small}
\begin{align*}
\Theta&= \frac{(\eta_--\eta_\star) (\eta_--\eta_+) \left(e^{\frac{2 (\eta_-+\eta_\star) \left(M-\Sigma
^2\right)}{\Sigma ^2}}-e^{\frac{2 (\eta_-+\eta_+) \left(M-\Sigma ^2\right)}{\Sigma ^2}}\right)}{(\eta_--\eta_\star) e^{\frac{2 (\eta_-+\eta_\star) \left(M-\Sigma ^2\right)}{\Sigma ^2}}+(\eta_+-\eta_-)
e^{\frac{2 (\eta_-+\eta_+) \left(M-\Sigma ^2\right)}{\Sigma ^2}}+(\eta_\star-\eta_+) e^{\frac{2 (\eta_\star+\eta_+) \left(M-\Sigma ^2\right)}{\Sigma ^2}}}\\&+\eta_\star+\frac{\Sigma ^2}{\Sigma
^2-M}+\eta_+.
\end{align*}
\end{small}
\end{proposition}
\begin{proof}
Let $f$ be the quadratic function
\[
f(\eta)=\frac{\eta(\eta-\eta_--\eta_\star)}{(\eta_+-\eta_\star)(\eta_+-\eta_-)}.
\]
As, $f(\eta_+)-f(\eta_\star)=1$, and $f(\eta_\star)-f(\eta_-)=0$, and 
\[
f'(\eta)=\frac{2\eta-\eta_--\eta_\star}{(\eta_+-\eta_-)(\eta_+-\eta_\star)}
\]
we get, by equation \eqref{eq: Ito max}, after division by $T$ and letting $T\rightarrow \infty$,
\[
\lim_{T\rightarrow\infty}\frac{\vert\{s<T: \Delta U_s>0\}\vert}{T}=\lim_{T\rightarrow\infty}\frac{1}{T}\int_0^T\mathcal A f(\eta_s)ds.
\]
It remains to evaluate the almost sure limit on the right side, by applealing to the Ergodic Theorem. Since
\[
\mathcal Af(\eta)=\frac{\Sigma^2}{(\eta_+-\eta_-)(\eta_+-\eta_\star)}+(M-\frac{\Sigma^2}{2})\frac{2\eta-\eta_--\eta_\star}{(\eta_+-\eta_-)(\eta_+-\eta_\star)}
\]
we get
\begin{align*}
\lim_{T\rightarrow\infty}\frac{1}{T}\int_0^T\mathcal A f(\eta_s)ds&=\frac{\Sigma^2-(M-\Sigma^2/2)(\eta_-+\eta_\star)}{(\eta_+-\eta_-)(\eta_+-\eta_\star)}\\&+\frac{2M-\Sigma^2}{(\eta_+-\eta_-)(\eta_+-\eta_\star)}\lim_{T\rightarrow \infty}\frac{1}{T}\int_0^T \eta_s ds.
\end{align*}
By Lemma \ref{lem: stationary distribution}, and using the stationary density \eqref{eq: stationary density maximal trades} of $\eta$, we get $\lim_{T\rightarrow \infty}\frac{1}{T}\int_0^T \eta_s ds=\int_{\eta_-}^{\eta_+}\eta \varphi(\eta)$
and thus \eqref{eq: max jumps exact}.

\end{proof}

We revisit the strategy that keeps the proportion of wealth approximately constant. The proof of the following is similar to that of Proposition \ref{prop: const prop minimal trades}, and thus omitted. Note that we do not give the explicit formula,
which can easily be derived from the previous statements, but would take too much place here. For the same reason, we only provide asymptotics of leading order.
\begin{proposition}\label{prop: bulky trades}
For $\Lambda\neq 0,1$ and $\mu\neq \sigma^2/2$, consider the strategy that buys at $\pi_-<\Lambda$ and sells at $\pi_+>\Lambda$ as much as necessary to bring the proportion of wealth in the risky asset $\pi_t$ back to $\Lambda$. If the trading boundaries are of the form
\begin{equation}\label{eq: ans boundaries pix} 
\pi_\pm=\Lambda \pm A_\pm \varepsilon^{1/3}+O(\varepsilon^{2/3}),
\end{equation}
with $A_\pm>0$, then average transaction costs satisfy for sufficiently small $\varepsilon$:
\begin{enumerate}
\item (Selling Frequency) The long-run selling frequency per unit time satisfies the asymptotics
\[
\text{SF}=\lim_{T\rightarrow\infty}\frac{\vert \{s\leq t: \Delta U_s>0\}}{T}=\frac{\sigma^2 \Lambda^2 (\Lambda-1)^2}{A_+(A_++A_-)}\varepsilon^{-2/3}+O(\varepsilon^{-1/3}).
\]
\item \label{item: 2x ATC} (Average Cost) Average transaction costs satisfy the asymptotics
\begin{equation}\label{eq: ATC GBM exactdl}
\text{ATC}:=\lim_{T\rightarrow\infty}\frac{\varepsilon}{T}\sum_{s\leq t} \pi_{t_-}\frac{\Delta \varphi^\downarrow_t}{\varphi_{t_-}}=\frac{\sigma^2 \Lambda^2 (\Lambda-1)^2}{(A_++A_-)}\varepsilon^{2/3}+O(\varepsilon)
\end{equation}
and thus is, at leading order, twice the cost of the control limit policy (that is, eq. \eqref{eq: ATC GBM asy} of 
Proposition \ref{prop: const prop minimal trades}).
\end{enumerate}

\end{proposition}

\subsubsection{Small Trades}\label{sec: micro trades}

The cumulative number of shares bought or sold lacks absolute continuity, necessitating approximation for real-world trade applications through trades of fixed size. However, executing discrete trades toward the midpoint of the no-trade region proves suboptimal, incurring twice the transaction costs compared to control limit policies (Proposition \ref{prop: bulky trades} \ref{item: 2x ATC}). Using the methods developed in this section, a similar outcome can be demonstrated for trades of any order comparable to the size of the no-trade region (approximately $\varepsilon^{1/3}$, as indicated in Remark \ref{rem: moderate trades}).

This prompts the question of determining the optimal trade size for a satisfactory approximation of the optimal trading policy. In the computation of Average Trading Costs, a crucial factor is the stationary density of the resetted or reflected diffusions (e.g., $\pi_t$). Hence, smaller trades contribute to a better approximation of the stationary density of the control limit policy, as illustrated in Figure \ref{fig: approaching control limit}.
\begin{figure}
\centering
\includegraphics[width=0.6\textwidth]{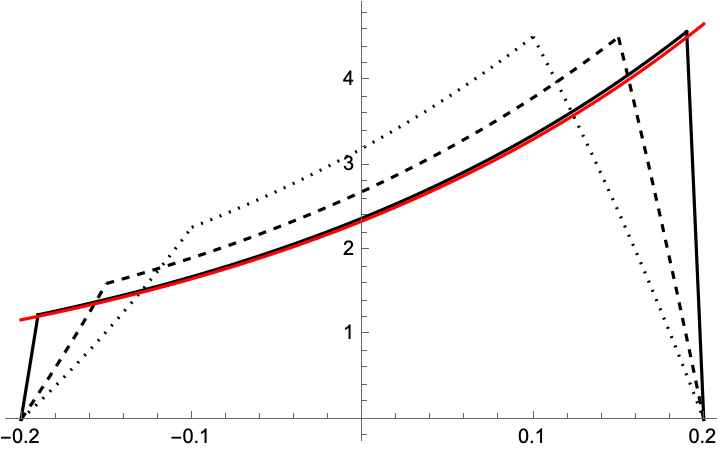}
\caption{The stationary densities converge, for smaller and smaller trades, to the stationary density of a reflected diffusion (red line). Here $\eta_\pm=\pm 0.2$, and
$\eta_{\pm,\star}=\pm 0.1$ (dotted line) and $\eta_{\pm,\star}=\pm 0.15$ (dashed line) and $\eta_{\pm,\star}=\pm 0.19$ (solid line). The parameters are $\mu=10\%$ and $\sigma=15\%$.}
\label{fig: approaching control limit}
\end{figure}
In this section, we substantiate this intuition by approximating the somewhat abstract reflected diffusion discussed in Section \ref{sec: minimal trades}. We achieve this by employing adequately small trades, ensuring that transaction costs align at the leading order with the control limit policy (Proposition \ref{prop: const prop minimal trades}).

This section is analogous to the previous one, with the modification that we trade from $\xi_-$ to $\xi_{-\star}$, and from $\xi_+$ to $\xi_{+,\star}$, where
\[
\xi_-<\xi_{-,\star}<\xi_{+,\star}<\xi_+.
\]
Define again $\eta=\log(\xi)$, then for any $C^2$ function $f$, we get the dynamics
\begin{align}\label{eq: Ito max1}
f(\eta_t)&=f(\eta_0)+\int_0^t \sigma f'(\eta_s)dW_s+\int_0^t \mathcal A f(\eta_s)ds\\\nonumber&+(f(\eta_{-,\star})-f(\eta_-))\vert \{s\leq t: \Delta L_s>0\}-(f(\eta_+)-f(\eta_{+,\star}))\vert \{s\leq t: \Delta U_s>0\}.
\end{align}
To extract the down-wards jumps $\Delta U_s$, we use $f(\eta)=(\eta-\eta_-)(\eta-\eta_{-,\star})$, for which $f(\eta_-)=f(\eta_{-,\star})$, and
\[
f(\eta_+)-f(\eta_{+,\star})=(\eta_+-\eta_-)(\eta_+-\eta_{-,\star})-(\eta_{+,\star}-\eta_-)(\eta_{+,\star}-\eta_{-,\star})>0.
\]
Furthermore,
\[
\mathcal Af(\eta)=\Sigma^2+(M-\Sigma^2/2)(2\eta-\eta_--\eta_{-,\star}).
\]
By dividing \eqref{eq: Ito max1} by $T$ and letting $T\rightarrow\infty$, we get
\begin{equation}\label{eq: superhit}
\lim_{T\rightarrow\infty}\frac{\vert \{s\leq t: \Delta U_s>0\}\vert}{T}=\frac{\lim_{T\rightarrow\infty}\frac{1}{T}\int_0^T\mathcal Af(\eta_s)ds}{(\eta_+-\eta_-)(\eta_+-\eta_{-,\star})-(\eta_{+,\star}-\eta_-)(\eta_{+,\star}-\eta_{-,\star})},
\end{equation}
a limit which exists almost surely, and can be computed by appealing to a suitably tailored ergodic theorem, knowing the stationary density. (The proof is similar to the one in Section \ref{sec: maximal trades}, but results in more tedious computations and longer proofs and thus is omitted.) A simulation in Figure \ref{fig: simulating density small bulk trades} demonstrates satisfactory agreement of the empirically computed stationary density of $\eta_T$ ($T=5$ years) with its exact counterpart $\eta$.

The sizes of the no-trade regions in all our examples are proportional to $\varepsilon^{1/3}$, thus sufficiently small trades must be of second order. We summarize trading frequency and
long-run transaction costs:\footnote{The majority of computations was performed by Wolfram Mathematica, using \eqref{eq: superhit}.}
\begin{proposition}\label{prop: micro trades}
For $\Lambda\neq 0,1$ and $\mu\neq \sigma^2/2$, consider the strategy that buys at $\pi_-<\Lambda$ resp.~ sells at $\pi_+>\Lambda$ as much as necessary to bring the proportion of wealth in the risky asset $\pi_t$ to $\pi_{+,\star}<\pi_+$ resp. ~ $\pi_{-\star}>\pi_-$. If the trading boundaries are of the form \eqref{eq: ans boundaries pix} with $A_\pm>0$, and if 
\begin{equation}\label{eq: tra bou best proxy}
\pi_{\pm,\star}=\pi_\pm\mp \kappa_\pm \varepsilon^{2/3}+O(\varepsilon),
\end{equation}
with $\kappa_\pm>0$, then or sufficiently small $\varepsilon$:
\begin{enumerate}
\item (Selling Frequency) The long-run selling frequency per unit time satisfies the asymptotics
\[
\text{SF}=\lim_{T\rightarrow\infty}\frac{\vert \{s\leq t: \Delta U_s>0\}\vert}{T}=\frac{\sigma^2 \Lambda^2 (\Lambda-1)^2}{2(A_++A_-)\kappa_+}\varepsilon^{-1}+O(\varepsilon^{-2/3}).
\]
\item (Average Cost) Average transaction costs satisfy the asymptotics
\begin{equation}\label{eq: ATC GBM exactdl}
\text{ATC}:=\lim_{T\rightarrow\infty}\frac{\varepsilon}{T}\sum_{s\leq t} \pi_{t_-}\frac{\Delta \varphi^\downarrow_t}{\varphi_{t_-}}=\frac{\sigma^2 \Lambda^2 (\Lambda-1)^2}{2(A_++A_-)}\varepsilon^{2/3}+O(\varepsilon)
\end{equation}
and thus agrees, at leading order, with the cost of the control limit policy (that is, eq. \eqref{eq: ATC GBM asy} of 
Proposition \ref{prop: const prop minimal trades}).
\end{enumerate}
\end{proposition}

\begin{figure}
\centering
\includegraphics[width=0.6\textwidth]{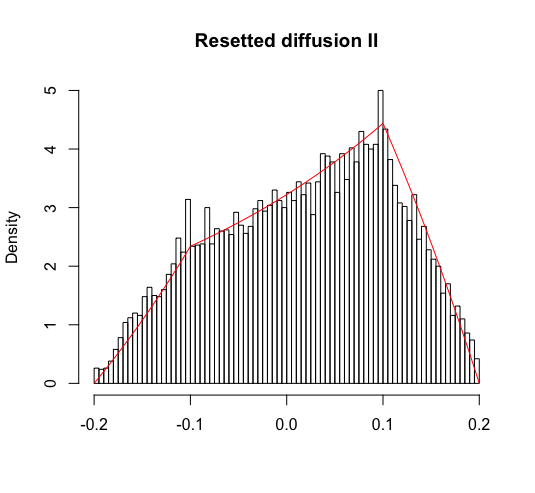}
\caption{Small trades of fixed size: For $M-\Sigma^2/2=10\%$ and $\Sigma=25\%$, and $\eta_-=-.2$, $\eta_l=-0.1$, $\eta_u=0.1$ and $\eta_+=+0.2$, we plot the stationary density of the resetted process
$\eta$. The histogram gives the value of $\eta_T$ for $T=5$ years, from ten thousand simulated path (daily frequency). }
\label{fig: simulating density small bulk trades}
\end{figure}

\begin{remark}\label{rem: moderate trades}
If, instead of \eqref{eq: tra bou best proxy}, one would use significantly larger trades, to
\[
\pi_{\pm,\star}=\pi_\pm\mp \kappa A_\pm \varepsilon^{1/3}+O(\varepsilon), \quad 0\leq\kappa<1,
\]
then selling frequency is given by
\[
\text{SF}=\lim_{T\rightarrow\infty}\frac{\vert \{s\leq t: \Delta U_s>0\}\vert}{T}=\frac{\sigma^2 \Lambda^2 (\Lambda-1)^2}{A_+(A_++A_-)(1-\kappa^2)}\varepsilon^{-2/3}+O(\varepsilon^{-1/3})
\]
and transaction costs are
\[
\text{ATC}:=\lim_{T\rightarrow\infty}\frac{\varepsilon}{T}\sum_{s\leq t} \pi_{t_-}\frac{\Delta \varphi^\downarrow_t}{\varphi_{t_-}}=\frac{\sigma^2 \Lambda^2 (\Lambda-1)^2}{(A_++A_-)(1+\kappa)}\varepsilon^{2/3}+O(\varepsilon).
\]
For $\kappa=1$, we recover the trades to some point $\Lambda$ in the no-trade region (Proposition \ref{prop: bulky trades}), and for any $\kappa\in (0,1)$, we get ATC slightly larger at leading order than the original one. This means that
such trading strategies do never match the trading costs of the control limit policy at the relevant, second order.

\end{remark}

\subsection{Tracking Frictionless Targets}\label{sec: tracking}
\subsubsection{Levered Funds}\label{sec: LETFS}
The manager of a Leveraged (or Inverse) Exchange Traded Fund (LIETF) has to maintain a constant exposure to a constant multiple $\Lambda$ of the benchmark's $S$ return with dynamics \eqref{eq: evolutionstheorie}.
For this section we take the conservative view that annualized average excess returns $\mu$ of $S$ vanish, which is line with the objective of a levered ETF manager to track a multiple of a benchmark's return rather than outperform it. \footnote{Besides, this assumption yields explicit formulas. Risk-premia, stochastic interest rates and stochastic volatility are included in the more general theory of performance evaluation of LIETFs in \cite{gm2023}.}

Due to the fact that frequent trading incurs transaction costs it is impossible to trade by perfectly matching the target leverage -- such a strategy would lead to immediate
bankruptcy, as it is of infinite variation. To measure the discrepancy between fund's excess returns (with wealth $w$, trading strategy $\varphi$), and the multiple $\Lambda$ of the underlying's excess return (the target), we use the (cumulative) difference $D_t$
\begin{equation}\label{cum diff}
dD_t = \frac{dw_t}{w_t} - r dt -\Lambda \left( \frac{dS_t}{S_t} - r dt \right),
\qquad D_0 = 0.
\end{equation}
The most relevant statistics of $D$, Tracking distance (TrD) and Tracking Error (TrE) are defined by
\begin{align}\label{eq: ExR}
\text{TrD}:=& \lim\sup_{T\rightarrow \infty}
\frac{1}{T}\mathbb E\left[\lim_{\Delta t\rightarrow 0}\sum_{k=1}^{\lfloor T/\Delta t \rfloor} (D_{k\Delta t} -D_{(k-1)\Delta t})\right]
=-\varepsilon \lim\inf_{T\rightarrow \infty}\frac{1}{T}\mathbb E\left[\int_0^T\pi_t\frac{d\varphi^\downarrow_t}{\varphi_t}\right], \\\label{eq: TrE}
\,\,\,\,\,\,
\text{TrE} :=& \,\,\,\, = \lim\inf_{T\rightarrow \infty}\frac{1}{T}
\left(\mathbb E\left[\lim_{\Delta t\rightarrow 0}\sum_{k=1}^{\lfloor T/\Delta t \rfloor} (D_{k\Delta t} -D_{(k-1)\Delta t})^2\right]\right)^{1/2}\,\,\,\,
\\&= \lim\inf_{T\rightarrow\infty}\left(\mathbb E\left[\frac1T \int_0^\infty (\pi_t-\Lambda)^2 \sigma_t^2 dt\right]\right)^{1/2} ,
\end{align}
with quadratic variation of $D$ being $\langle D\rangle_T$. 

The explicit expression on the right side of \eqref{eq: ExR} reveals that negative Tracking distance is equal to Average Tracking Costs (ATC) (cf.~\cite{gm2020}) . Thus, managing an LIETF means to find an optimal trade-off between Tracking Distance TrD (or ATC) and Tracking Error TrE (for a more detailed motivation of this objective, cf.~\cite{gm2023}).

Maximizing Tracking Distance for a fixed Tracking Error can be summarized in minimizing the long-run Equivalent Expense Ratio (EER)
\begin{equation}\label{eq: obj asympt}
\text{EER}(\varphi) :=\frac{\gamma}{2}\tracking^2-\text{TrD} =\liminf_{T\rightarrow\infty} 
\frac{1}{T} \mathbb E\left[\frac{\gamma}{2}\int_0^T (\pi_t -\Lambda)^2\sigma^2 dt
+\varepsilon\int_0^T\pi_t\frac{d\varphi^\downarrow_t}{\varphi_t} 
\right],
\end{equation}
where the Lagrange multiplier $\gamma>0$ is interpreted as aversion to tracking error.

By \cite[Theorem 3.1]{gm2023}, for sufficiently small $\varepsilon$, the problem
\[
\max \text{EER}(\varphi)
\]
over all admissible strategies $\varphi\in \Phi$, is well posed: There exists $0<\pi_-<\pi_+<\infty$ such that the trading strategy $\hat\varphi$ that buys at $\pi_- $ and sells at $\pi_+ $ to keep the risky weight $\pi_t$ within the interval $[\pi_-,\pi_+]$ is optimal. Furthermore, the minimum \emph{Equivalent Expense Ratio} is 
\begin{equation*}
\text{EER}=\frac{\gamma\sigma^2}{2}(\pi_- -\Lambda)^2,
\end{equation*}
the optimal \emph{Tracking Difference} and \emph{Tracking Error} are 
\begin{align}\label{eq1sup}
\tdiff &:=\lim_{T\rightarrow\infty}\frac{D_T}T = -\frac{\sigma^2}{2}\frac{\pi_-\pi_+(\pi_+-1)^2}{(\pi_+-\pi_-)(1/\varepsilon-\pi_+)},\\
\label{eq: Tr.E.}
\tracking &:= \lim_{T\rightarrow\infty} \sqrt{\frac{\langle D\rangle_T}T }=
\sigma\sqrt{(\pi_--\Lambda)^2-\gamma\frac{\pi_-\pi_+(\pi_+-1)^2}{(\pi_+-\pi_-)(1/\varepsilon-\pi_+)}}.
\end{align}

The trading boundaries have the series expansion
\begin{equation}\label{as multiplier piminust}
\pi_\pm = \Lambda 
\pm\left(\frac{3}{4\gamma}\Lambda^2(\Lambda-1)^2\right)^{1/3}\varepsilon^{1/3} 
-\frac{\Lambda}{\gamma}\left(\frac{\gamma \Lambda (\Lambda-1)}{6}\right)^{1/3}\varepsilon^{2/3} +O(\varepsilon).
\end{equation}
In this section, we put forward an approximation method that allows to obtain the above trading boundaries and performance measures, at high precision. First, we study the shadow market approach, and second, the associated HJB equations.
\subsubsection*{The Shadow Market}\label{sec: shadow LETF}
We assume that a Shadow price is a continuous semimartingale with absolutely continuous differential characteristics, that is,
\[
\frac{d\widetilde S_t}{\widetilde S_t}=rdt+\widetilde\mu_tdt+\widetilde \sigma_t dB_t.
\]

In the shadow market, the cumulative tracking difference $\widetilde D_t$ is of the form
\[
d\widetilde D_t = \frac{d\widetilde F_t}{\widetilde F_t} - r dt -\Lambda \left( \frac{d\widetilde S_t}{\widetilde S_t} - r dt \right),
\]
which results in the shadow equivalent expense ratio
\[
\widetilde {\text{EER}}(\varphi)=\liminf_{T\rightarrow\infty} \left(
\frac{1}{T} \mathbb E\left[\frac{\gamma}{2}\int_0^T (\widetilde\pi_t -\Lambda)^2\widetilde \sigma_t^2 dt\right]-\frac{1}{T} \mathbb E\left[\int_0^T (\widetilde\pi_t -\Lambda)\widetilde \mu_t^2 dt\right]\right).
\]

This objective is optimized by trading so that, at all times, the fraction of wealth in $\widetilde S$ equals the Merton fraction 
\begin{equation}\label{eq: merton shadow}
\widetilde \pi_t=\Lambda+\frac{\widetilde \mu_t}{\gamma \widetilde\sigma_t^2}.
\end{equation}
Assume a shadow price of the form $\widetilde S_t/S_t=g(\pi_t)$, where the function $g$ is assumed to be twice differentiable, satisfying $g(\pi_-)=1$, $g(\pi_+)=(1-\varepsilon)$, and the smooth pasting condition $g'(\pi_\pm)=0$. This implies that
\begin{equation}\label{sh 2x}
\widetilde{\mu}_t=\frac{g'(\pi_t)(\pi_t(1-\pi_t)^2\sigma^2)+\frac{1}{2}g''(\pi_t)\pi_t^2(1-\pi_t)^2\sigma^2}{g}
\end{equation}
and
\begin{equation}\label{eq: diff coefx}
\widetilde{\sigma}_t=(\sigma g+g'(\pi)\pi(1-\pi)\sigma)/g.
\end{equation}
The proportion of wealth in the shadow asset is
\begin{equation}\label{eq: shadow prop}
\widetilde \pi_t=\frac{g(\pi_t) \pi_t}{g(\pi_t) \pi_t+(1-\pi_t)}.
\end{equation}
By equating \eqref{eq: shadow prop} with \eqref{eq: merton shadow}, using \eqref{sh 2x}--\eqref{eq: diff coefx}, we obtain the following
free boundary problem for $g$,
\begin{equation}\label{eq: fbp psi1}
\frac{1}{2} g'' \pi^2(1-\pi)^2+g' \pi (1-\pi)^2= \gamma (g+g' \pi (1-\pi))^2\left(\frac{g \pi }{g \pi+(1-\pi)}-\Lambda\right),
\end{equation}
subject to \eqref{eq: shadow constraints}--\eqref{eq: smooth pasting}.

In view of formula \eqref{as multiplier piminust} and the results in \cite{em_1_2024}, we conjecture that the trading boundaries agree for small transaction costs with \eqref{as multiplier piminus} at first order; furthermore, we anticipate a minor discrepancy in the second order terms. Therefore, we conjecture that
the free boundaries are of the form
\begin{equation}\label{eq: shadow trading boundaries}
\widetilde\pi_\pm = \Lambda 
\pm\left(\frac{3}{4\gamma}\Lambda^2(\Lambda-1)^2\right)^{1/3}\varepsilon^{1/3} 
-\kappa_\pm\frac{\Lambda}{\gamma}\left(\frac{\gamma \Lambda (\Lambda-1)}{6}\right)^{1/3}\varepsilon^{2/3} +c_\pm \varepsilon+O(\varepsilon^{4/3})
\end{equation}
with unknown constant $c_\pm$ and multiplier $\kappa_\pm$. Furthermore,
being depraved of an explicit solution of the ODE \eqref{eq: fbp psi1}, we make an Ansatz for a polynomial approximation
\begin{align}\label{eq: proxy g}
g(z)&=\left(1 + C_- (z - \widetilde\pi_-)^3 + D_- (z - \widetilde\pi_-)^4 + 
E_- (z - \widetilde\pi_-)^5)\right) \frac{z - \widetilde\pi_+}{\widetilde\pi_- - \widetilde\pi_+}\\&+ \left(1 - \delta^3
+ C_+ (z - \widetilde\pi_+)^3 + D_+ (z - \widetilde\pi_+)^4 + 
E_+ (z - \widetilde\pi_+)^5)\right) \frac{z - \widetilde\pi_-}{\widetilde\pi_+ - \widetilde\pi_-} \\\nonumber&+ 
M_0 (z - \widetilde\pi_-)^2 (z - \widetilde\pi_+)^2 + 
M_+ (z - \widetilde\pi_-) (z - \widetilde\pi_+)^3 + M_- (z - \widetilde\pi_-)^3 (z - \widetilde\pi_+) \\\nonumber&+ 
G_+ (z - \widetilde\pi_-) (z - \widetilde\pi_+)^2 + G_- (z - \widetilde\pi_+) (z - \widetilde\pi_-)^2 \\\nonumber&+ 
L (z - \widetilde\pi_-) (z - \widetilde\pi_+)^4 + L (z - \widetilde\pi_-)^4 (z - \widetilde\pi_+),
\end{align}
where the choice $\delta=\varepsilon^{1/3}$, acknowledges that the expansions of trading boundaries \eqref{eq: shadow trading boundaries}
are formal power series in $\delta$. 

Next, we sketch how to identify the unknown constants in \eqref{eq: shadow trading boundaries} and \eqref{eq: proxy g}. By construction, $g(\widetilde\pi_-)=1$ and $g(\widetilde\pi_+)=1-\varepsilon$ at third order. To satisfy $g'(\widetilde \pi_+)=0$ at second order, we must set
\[
G_+=G_-
\]
and
\[
C_-=-\frac{\gamma-6 G_+(\Lambda-1)^2\Lambda^2}{6(\Lambda -1)^2\Lambda^2}
\]
and for having $g'(\widetilde \pi_-)=0$ up to second order, we need to insist on $C_+=C_-$. Next,
if $g$ satisfies the ODE \eqref{eq: fbp psi1} at $\widetilde\pi_\pm$, we get the coefficients $D_\pm$
(lengthy expressions in terms of $\kappa_\pm$, $M_0$, $M_\pm$) and they satisfy 
\[
D_+-D_-=-\frac{3}{4}(M_+-M_-).
\]
If we insist on $g'(\widetilde\pi_\pm)=0$ at third order, we get $M_\pm$ as functions of
$M_\mp$, $M_0$, and $\kappa_\pm$. At the moment, the ODE is satisfied at second order only at
$\widetilde \pi_\pm$. We finally insist that the ODE is asymptotically satisfied at any $z$ in the interval, which can be represented as a convex combination of
the free boundaries,
\[
z=\theta \widetilde \pi_-+(1-\theta)\widetilde \pi_+,\quad 0\leq \theta\leq 1.
\]
In the interior of the interval, $g$ satisfies the ODE at second order, only if 
\[
\kappa_\pm=-1.
\]
Furthermore, to have the ODE satisfied at third order, we get \[c_++c_-=\Lambda(1-\Lambda),\]
\[
E_+=E_-=\frac{\gamma(-99+2(153+33\gamma(\Lambda-1)-137\Lambda)\Lambda)}{360\Lambda^4(\Lambda-1)^4},
\]
and
\[
L_+=\frac{\gamma(-9+(30+9\gamma(\Lambda-1)-29\Lambda)\Lambda)}{36\Lambda^4(\Lambda-1)^4},
\]
\[
D_-=M_0-\frac{\gamma(3+(\Lambda(\kappa_+-2(3+\kappa_-))}{6\Lambda^3(\Lambda-1)^3}
\]
and, similarly if $g$ satisfies the ODE \eqref{eq: fbp psi1} at $\pi_+$, we get $D_+=D_-$. These computations are summarized in the following result:
\begin{proposition}
The optimal strategy in the market of the frictionless asset $\widetilde S$ is a control limit policy for some trading boundaries $\widetilde \pi_\pm$. For sufficiently small $\varepsilon$, they satisfy the asymptotic expansion
\begin{equation}\label{tata}
\widetilde\pi_\pm = \Lambda 
\pm\left(\frac{3}{4\gamma}\Lambda^2(\Lambda-1)^2\right)^{1/3}\varepsilon^{1/3} 
+\frac{\Lambda}{\gamma}\left(\frac{\gamma \Lambda (\Lambda-1)}{6}\right)^{1/3}\varepsilon^{2/3} +O(\varepsilon^{4/3})
\end{equation}
and thus agree with the optimal policy in the original market\eqref{as multiplier piminust}
only at first order. \footnote{A comparison with \eqref{as multiplier piminust} shows that signs of the second order coefficients are switched, which is analogous to the sign switch in the shadow optimal policy for local mean-variance, cf. \eqref{as multiplier piminus} and the discussion thereafter.}
\end{proposition}

\subsubsection*{The Original Market}

Due to \cite[Theorem 3.1]{gm2023} for small enough bid-ask spread $\varepsilon$ there exist some trading boundaries $\pi_-<\pi_+$ such that the policy $\hat\varphi$ which buys minimally at $\pi_- $ and sells minimally at $\pi_+ $ so to maintain the fraction of wealth $\pi_t$ between $\pi_-$ and $\pi_+$ is optimal. 

The trading boundaries are obtained in terms of the risk-safe ratio $\zeta$. These are of the form $\pi_\pm=\frac{\zeta_\pm}{1+\zeta_\pm}$, where $\zeta_\pm$ are part of
the solution to the free boundary problem for the triple ($W$, $\zeta_\pm$), with a scalar function $W$ satisfying
\begin{align}\label{eq: TKA fbp1}
&\frac{1}{2} \zeta^2 W''(\zeta)+\zeta W'(\zeta)-\frac{\gamma}{(1+\zeta)^2}\left(\Lambda-\frac{\zeta}{1+\zeta}\right)=0,\\\label{initial0 TKA1}
&W(\zeta_-)=0,\qquad W'(\zeta_-)=0,\\\label{terminal0 TKA1}
&W(\zeta_+)=\frac{\varepsilon}{(1+\zeta_+)(1+(1-\varepsilon)\zeta_+)}, \quad W'(\zeta_+)=\frac{\varepsilon(\varepsilon-2(1-\varepsilon)\zeta_+-2)}{(1+\zeta_+)^2(1+(1-\varepsilon)\zeta_+)^2},
\end{align}
which has a unique solution for which $\zeta_-<\zeta_+$. To allow an adaption of the method of Section \ref{sec: shadow LETF}, we use the transformation
\begin{equation}\label{eq: trafo}
\frac{\zeta g(\pi(\zeta))}{1+\zeta g(\pi(\zeta))}=\frac{\zeta}{1+\zeta}-\zeta W(\zeta),\quad \pi(\zeta)=\frac{\zeta}{1+\zeta}
\end{equation}
which translates \eqref{eq: TKA fbp1}--\eqref{terminal0 TKA1} into the following problem for ($g=g(\pi)$, $\pi_\pm$)
\begin{tiny}
\begin{align*}\nonumber
&2 \gamma (\pi (g(\pi )-1)+1)^3 (\pi -\Lambda )+2 \pi (\pi (g(\pi )-1)+1) \left(-g'(\pi )+\pi \left(g'(\pi )+g(\pi )^2-2 g(\pi )+1\right)+g(\pi )^2-1\right)\\&-\pi
^2 \bigg(g''(\pi )-4 g(\pi ) g'(\pi )-2 g'(\pi )+\pi ^2 \left(3 g''(\pi )+4 g'(\pi )^2-2 g'(\pi )+2 g(\pi )
\left(-g''(\pi )+g'(\pi )+3\right)+2 g(\pi )^3-6 g(\pi )^2-2\right)\\&+\pi ^3 \left((g(\pi )-1) g''(\pi )-2 g'(\pi
)^2\right)+\pi \left(-3 g''(\pi )-2 g'(\pi )^2+4 g'(\pi )+g(\pi ) \left(g''(\pi )+2 g'(\pi )-6\right)+2 g(\pi
)^3+4\right)+2 g(\pi )^3-2\bigg)=0
\end{align*}
\end{tiny}
subject to \eqref{eq: shadow constraints}--\eqref{eq: smooth pasting}.

We use again the general Ansatz for $g$ in \eqref{eq: proxy g} and the Ansatz \eqref{eq: shadow trading boundaries} for the trading boundaries, relabelling them as $\pi_\pm$. Since the boundary conditions are the same as in Section \ref{sec: shadow LETF}, we get the exact same coefficients $G_+=G_-$ and $C_+=C_-$.

Since the ODE is different, we have to compute the remaining coefficients again: If $g$ satisfies the ODE \eqref{eq: fbp psi1} at $\pi_\pm$, we get the coefficients $D_\pm$
(lengthy expressions in terms of $\kappa_\pm$, $M_0$, $M_\pm$, but different to the Shadow market) and they satisfy also
\[
D_+-D_-=-\frac{3}{4}(M_+-M_-).
\]
If we insist on $g'(\pi_\pm)=0$ at third order, we get $M_\pm$ as functions of $M_\mp$, $M_0$, and $\kappa_\pm$. If we finally insist on the ODE be satisfied at any $z$ in the no-trade region, then $g$ satisfies the ODE only if 
\[
\kappa_\pm=1.
\]
This sign is opposite to the Shadow market solution. This settles an alternative proof to \cite{gm2023} of the asymptotics of the optimal policy \eqref{as multiplier piminust}. We finally discuss transaction costs:
\begin{remark}\label{rem: LETF ATC}
For the trading boundaries \eqref{as multiplier piminust} (optimal in original market) and \eqref{tata} (optimal in potential shadow market) we get transaction costs
\[
\text{ATC}=\frac{3 \sigma^2}{\gamma}\left(\frac{\gamma\Lambda (\Lambda-1)}{6}\right)^{4/3}\varepsilon^{2/3}+\frac{\sigma^2}{2}\Lambda^2(\Lambda-1)\varepsilon+O(\varepsilon^{4/3}).
\]
Note that this is precisely the asymptotic expansion of the negative tracking difference \eqref{eq1sup}, which may found in \cite[Theorem 3.1, eq.(17)]{gm2023}. 
\end{remark}

\subsubsection{Neuberger's Log Contract}\label{sec: log contract}
In a frictionless market with a risky asset $S$ that is a continuous semimartingale, and cash earning zero interest ($r=0$),\footnote{A small modification for non-zero interest rate
applies.} a variance swap can be perfectly replicated by a combination of
a static hedge in options to replicate the Log contract, and a dynamic hedge with an exposure in a risky position $Y_t$ constant at all times to some $y_\star>0$. Such strategy results in insolvency under proportional costs.

Let us again assume that the ask-price is given by \eqref{eq: evolutionstheorie}, and proportional costs are equal to $\varepsilon\in (0,1)$. Because investment opportunities are constant,
the variance swap is trivial, with vanishing net cashflow upon maturity $T$. Since
\[
\log(S_T/S_0)=\int_0^T \frac{dS_t}{S_t}-\frac{\sigma^2 T}{2},
\]
we consider instead dynamic hedging $y_\star>0$ units of the log-contract that pays $\log(S_T/S_0)$ at time $T>0$.\footnote{In practice this should be done by trading futures on $S$, tied to the log-contract, cf.~\cite{neuberger}. } This is a non-trivial question in itself - as it require to modify the frictionless strategy $Y_t=y_\star$ to remain solvent. This task can be attempted by finding a non-anticipating, finite-variation strategy that minimizes the squared hedging error
\begin{equation}\label{eq: tre hedge}
\frac{1}{T}\mathbb E\left[\int_0^T \langle (Y_t-y_\star) \frac{dS_t}{S_t}\rangle_T\right] =\frac{1}{T}\mathbb E\left[\int_0^T (Y_t-y_\star)^2\sigma^2dt\right]
\end{equation}
while keeping annualized transaction costs 
\begin{equation}\label{eq: trc hedge}
\frac{\varepsilon}{T}\mathbb E\left[\int_0^T Y_t \frac{d\varphi_t^{\downarrow}}{\varphi_t}\right]
\end{equation}
acceptable. With $\gamma>0$, the aversion to hedging error, we minimize 
\begin{equation}\label{eq: obj vs}
\frac{\gamma\sigma^2}{2} \frac{1}{T}\mathbb E\left[\int_0^T (Y_t-y_\star)^2dt\right]+\frac{\varepsilon}{T}\mathbb E\left[\int_0^T Y_t \frac{d\varphi_t^{\downarrow}}{\varphi_t}\right].
\end{equation}

Please note that even when transaction costs are zero, a self-financing strategy that maintains $Y_t$ at a constant value of $y_\star$ may with low probability necessitate excessively large funds by time $T$ to cover rebalancing costs, rendering the strategy potentially insolvent. This issue arises because the fair value of a log-contract is given by
\[
L_t=\log(S_t)-\frac{\sigma^2}{2}(T-t)
\]
thus being a Brownian motion with drift. This challenge is addressed by entering in futures positions on the underlying $S$ rather than trading $S$ itself and linking the futures position to the shorted log-contract, which is also futures traded with $L_t$ as settlement price. This results in changes in the futures positions being immediately offset by the log-contract, providing a perfect hedge (cf.~\cite[p.~78]{neuberger}. In the following we only implicitly adopt these mechanics, by ignoring the solvency issue.

For tractability reasons, we resort to an infinite time horizon, and thus minimize
\begin{equation}\label{eq: obj vsx}
\frac{\gamma\sigma^2}{2} \lim_{T\rightarrow\infty}\frac{1}{T}\mathbb E\left[\int_0^T (Y_t-y_\star)^2dt\right]+\varepsilon\lim_{T\rightarrow\infty}\frac{1}{T}\mathbb E\left[\int_0^T Y_t \frac{d\varphi_t^{\downarrow}}{\varphi_t}\right]
\end{equation}
over all admissible strategies. For time horizons such as five years, and daily hedging, the strategies
typically approximate well those for the finite horizon objective \ref{eq: obj vs}. (cf. ~the robustness checks in \cite[Section 6.2 and Figure 6]{gm2023}). For shorter maturities, such as those aligning with the expiry of European calls or puts, a slight adjustment of the analysis may be necessary. This adjustment yields time-dependent strategies instead of the steady-state solutions discussed earlier. The rigorous treatment of finite-time horizon remains a subject for future research.

We introduce now the hedging exercise. Knowing the perfect target, we conjecture that it is optimal to not engage in trading, as long as $Y_t$ ranges in $[y_-,y_+]$, where $y_-<y_\star<y_+$, and to trade minimally at the boundaries $y_\pm$,
so to keep the risky position $Y$ inside the interval. In other words, $Y$ is a stationary process, realised by the reflected diffusion
\begin{equation}\label{eq: toto}
dY_t=Y_t \left(\frac{d\varphi^\uparrow_t}{\varphi_t}-\frac{d \varphi^\downarrow_t}{\varphi_t}+\frac{dS_t}{S_t}\right)=Y_t\left(dL_t-dU_t+\mu dt+\sigma dB_t\right),
\end{equation}
with $\varphi^\uparrow$ being supported on $Y_t=y_-$, and $\varphi^\downarrow$ being supported on $Y_t=y_+$. In a similar manner as Lemma \ref{lem: existence controllable strategy}
we get
\begin{lemma}\label{lem: existence controllable strategyx}
For $0<y_-<y_+$, there exists an admissible trading strategy ${\hat\varphi}$ such that the risky position $Y_t$
satisfies SDE \eqref{eq: toto}. Moreover, $(Y_t,{\hat\varphi}_t^\uparrow,{\hat\varphi}_t^\downarrow)$ is a reflected diffusion on
the interval $[y_-,y_+]$. In particular, $Y_t$ has stationary density equals
\begin{equation}\label{eq st de 1}
\nu(y):=\frac{\frac{2\mu}{\sigma^2}-1}{y_+^{\frac{\mu}{\sigma^2}-1}-y_-^{\frac{2\mu}{\sigma^2}-1}}y^{\frac{2\mu}{\sigma^2}-2}, \quad y\in[y_-,y_+].
\end{equation}
\end{lemma}
Next, we give a heuristic derivation of the HJB equations, that will allow to identify an optimal strategy. While the heuristics provided here to derive the free boundary problem follow the arguments of \cite{gm2020} for local-mean variance, the resulting equations are different. 

For finite-horizon $T$, we have the objective
\begin{equation}
\max_{\varphi\in \Phi}\mathbb E\left[\gamma\sigma^2\int_0^T \left(y_\star Y_t-\frac{Y_t^2}{2}\right)dt-\varepsilon\int_0^TY_t\frac{d\varphi^\downarrow_t}{\varphi_t} \right].
\end{equation}
Assuming that the value function $V$ depends on $Y_t$ and time $t$ only, the objective's conditional value is
\begin{equation}\textstyle
\textstyle{F^\varphi(t) = \gamma\sigma^2\int_0^t \left(y_\star Y_s-\frac{Y_s^2}{2}\right)ds-\varepsilon\int_0^t Y_s\frac{d\varphi^\downarrow_s}{\varphi_s} 
+ V(t,Y_t)
.}
\end{equation}
Applying the It\^o formula, we get 
\begin{align*}
dF^\varphi(t) &= \gamma\sigma^2\left(y_\star Y_t -\frac{Y_t^2}{2}\right)dt -\varepsilon Y_t\frac{d\varphi^\downarrow_t}{\varphi_t} +V_tdt + \partial_y V dY_t + \frac12 \partial ^2_y V\langle Y_t \rangle_t,
\end{align*}
and thus $F^\varphi$ satisfies the dynamics
\begin{align}
\textstyle{dF^\varphi(t) =}&\textstyle{ 
\left(\gamma\sigma^2y_\star Y_t -\frac{\gamma \sigma^2}{2}Y_t^2
+V_t + \frac{\sigma^2}2 Y_t^2 \partial^2_y V_y + (\mu+r) Y_t \partial_y V
\right)dt}\\
-& Y_t \left(\partial_y V + \varepsilon \right)\frac{d\varphi^\downarrow_t}{\varphi_t} 
+ Y_t \partial_y V \frac{d\varphi^\uparrow_t}{\varphi_t}
+\sigma Y_t \partial_yV dW_t.
\end{align}

The martingale principle of optimal control implies that the process \(F^\varphi(t)\) is a supermartingale for any strategy \(\varphi\) but a martingale for the optimal one. Thus
\begin{align}\label{eq:boucon}
-\varepsilon \leq V_y \leq 0.
\end{align}

Similarly, it follows that
\begin{equation*}
\gamma\sigma^2 y_\star y-\frac{\gamma \sigma^2}{2}y^2
+V_t + \frac{\sigma^2}2y^2 \partial^2_y V + \mu y \partial_y V \leq 0,
\end{equation*}
In order to obtain a steady-state solution, we assume that \(V(t,y) = \kappa (T-t) - \int^y W(z) \, dz\) for a constant \(\kappa\) which represents the average long-run performance. With this assumption, the previous inequalities turn into
\begin{align}
0 \leq W(y) \leq \varepsilon, \\
\gamma \sigma^2 y y_\star -\frac{\gamma \sigma^2}{2}y^2
-\kappa - \frac{\sigma^2}2 y^2 W'(y) - \mu y W(y) \leq 0.
\end{align}
We conjecture now, similar as in \cite{gm2020} that the first is an inequality on some interval \([y_-,y_+]\), which turns into equality at the boundary of the interval, and thus
\begin{align}\label{eq:hjb}
\frac{\sigma^2}2 y^2 W'(y) + \mu y W(y) 
-\gamma\sigma^2 y_\star y+\frac{\gamma \sigma^2}{2}y^2
+\kappa = 0 \qquad \text{for }y \in [y_-,y_+],\\
\label{eq:bou1}
W(y_-) = 0, \quad W(y_+) = \varepsilon.
\end{align}
This problem is a second-order ODE with two free boundaries, and two unknown constants. To identify the latter, we use the smooth pasting conditions
\begin{equation}\label{eq:smooth1}
W'(y_-) = 0, \quad W'(y_+) = 0.
\end{equation}

Let $h(y):=\gamma \sigma^2 y_\star y-\frac{\gamma \sigma^2}{2}y^2$. The initial conditions $W(y_-)=W'(y_-)=0$ identify $\kappa=h(y_-)$,
and thus the ODE becomes
\[
\frac{\sigma^2}2 y^2 W'(y) + \mu y W(y) =h(y)-h(y_-).
\]
The solution of the corresponding initial value problem (which can be found by variation of constants) is given by
\[
W(y)=\frac{2}{\sigma^2 y^\alpha}\int_{y_-}^y (h(u)-h(y_-))u^{\alpha-2}du,\quad\text{where}\quad\alpha=\frac{2\mu}{\sigma^2}.
\]
Insisting on the terminal conditions, $W(y_\pm)=\varepsilon$ and $W'(y_\pm)=0$, we get the following system of non-linear equations for ($y_-,y_+$),
\begin{align}\label{eq: nl1}
\int_{y_-}^{y_+} (h(u)-h(y_-))u^{\alpha-2}du&=\frac{\varepsilon \sigma^2}{2}y_+^{\alpha},\\\label{eq: nl2}
\int_{y_-}^{y_+} (h(u)-h(y_-))u^{\alpha-2}du&=\frac{y_+^{\alpha-1}}{\alpha}(h(y_+)-h(y_-)).
\end{align}

\begin{lemma}\label{lem: y asy}
For sufficiently small $\varepsilon$, the system \eqref{eq: nl1}--\eqref{eq: nl2} has a unique solution $y_\pm$, with asymptotics \eqref{eq: as ypm} satisfying
\begin{equation}\label{eq: series expansions}
y_\pm=y_\star\pm \left(\frac{3 y_\star^2}{4\gamma}\right)^{1/3} \varepsilon^{1/3}-\left(\frac{\gamma y_\star}{6}\right)^{1/3}\frac{\mu}{\gamma \sigma^2}\varepsilon^{2/3}+ O(\varepsilon).
\end{equation}
\end{lemma}
\begin{proof}
We set $\delta=\varepsilon^{1/3}$. We make the following Ansatz for the asymptotic expansions of the boundaries
\begin{equation}\label{eq: as ypm}
y_\pm=y_\star\pm A_\pm \varepsilon^{1/3}-B_\pm \varepsilon^{2/3}+O(\varepsilon).
\end{equation}
For $\delta\rightarrow 0$, the integral on the right hand side of the equations vanishes. We develop this integral asymptotically for small $\delta$. To this end, observe that
\[
\frac{d}{d\delta} \int_{y_-}^{y_+} (h(u)-h(y_-))u^{\alpha-2}du=y_+^{\alpha-2}(h(y_+)-h(y_-))\frac{d y_+}{d\delta }-h'(y_-)\left(\frac{y_+^{\alpha-1}}{\alpha-1}-\frac{y_-^{\alpha-1}}{\alpha-1}\right)\frac{dy_-}{d\delta}.
\]
Developing this expression as a formal power series in $\delta$ and then integrating once, yields the asymptotic expansion of the left side of the equations \eqref{eq: nl1}--\eqref{eq: nl2}. (Alternatively, once can integrate explicitly. The present method has the advantage to be applicable to general functions $h$, not just the quadratic one here.) Developing also
the right sides, shows that the first equation's leading term is of third order, and the second one has a second order term. To make them vanish, we get the following system of equations,
\begin{align*}
(2 A_--A_+)(A_-+A_+)^2\gamma&=3 y_\star^2,\\
A_-^2&=A_+^2
\end{align*}
which implies
\[
A_+=A_-=\left(\frac{3 y_\star^2}{4\gamma}\right)^{1/3}.
\]
Next, to satisfy \eqref{eq: nl1} at fourth order, we equate the fourth order term to zero, obtaining
\[
B_-=\left(\frac{\gamma y_\star}{6}\right)^{1/3}\frac{\mu}{\gamma \sigma^2},
\]
and to satisfy \eqref{eq: nl2} at third order, we must have $B_+=B_-$. This establishes the asymptotics \eqref{eq: series expansions}.
\end{proof}

We now proceed to the asymptotic method, to provide an

\begin{proof}[Alternative Proof of Lemma \ref{lem: y asy}] 
For \(g=1-W\) the free boundary problem \eqref{eq:hjb}--\eqref{eq:smooth1} takes the form 
\begin{align}\label{eq: deq y}
\frac{1}2 y^2 g''(y) + (2\alpha+1) y g'(y) +2\alpha g(y)
+\gamma(y_\star -y)-2\alpha &=0,\\
g(Y_-)=1,\quad g'(Y_+)&=1-\varepsilon,\\
g'(Y_-)=0,\quad g'(Y_+)&=0,
\end{align}
which is essentially of the form \eqref{eq: ODE2}--\eqref{eq: smooth pasting} except that \(g\) is a function of the risky position \(Y\), instead of \(\pi\). We use an Ansatz similar to \eqref{eq: proxy g} for \(g\) however with \(L_+\neq L_-\) (otherwise the ODE cannot be satisfied at third order in the interior of the interval), but assume, \(G_+=G_-\) (guided by previous two problems). Noting that the variable \(z\) is interpreted as the risky position \(y\), rather than its proportion of wealth:
\begin{align}\label{eq: proxy gy}
g(y)&=\left(1 + C_- (y - y_-)^3 + D_- (y - y_-)^4 + 
E_- (y -y_-)^5)\right) \frac{y - y_+}{y_--y_+}\\&+ \left(1 - \delta^3
+ C_+ (y -y_+)^3 + D_+ (y-y_+)^4 + 
E_+ (y-y_+)^5)\right) \frac{y-y_-}{y_+ -y_-} \\\nonumber&+ 
M_0 (y-y_-)^2 (y-y_+)^2 + 
M_+ (y-y_-) (y-y_+)^3 + M_- (y-y_-)^3 (y-y_+) \\\nonumber&+ 
G_+ (y-y_-) (y-y_+)^2 + G_+ (y-y_+) (y-y_-)^2 + \\\nonumber&+
L_+ (y -y_-) (y_-y_+)^4 + L_- (y-y_-)^4 (y-y_+).
\end{align}
Moreover, we use the Ansatz \eqref{eq: as ypm} for the trading boundaries, assuming knowledge of the first order terms.\footnote{The approximation can be performed without knowing the coefficients \(A_\pm\). Also, one can obtain higher order coefficients with this method.} Being thus spoiled, 
we only need to identify the coefficients \(B_\pm\) in
\begin{equation}\label{eq: series expansionsx}
y_\pm=y_\star\pm \left(\frac{3 y_\star^2}{4\gamma}\right)^{1/3} \varepsilon^{1/3}-B_\pm \varepsilon^{2/3}+ O(\varepsilon).
\end{equation}
By construction, \(g(y_-)=1\) and \(g(y_+)=1-\varepsilon\) at third order. To satisfy \(g'(y_+)=0\) at second order, we must set
\[
C_- = G_+ -\frac{\gamma}{6 y_\star ^2}.
\]
Similarly, to have \(g'(y_-)=0\) at second order, we need \(C_+=G_+\).
Next,
if \(g\) satisfies the ODE \eqref{eq: deq y} at \(y_+\), we get 
\[
D_+=\frac{1}{48}\left(12 M_0+36M_++\frac{6\gamma}{y_\star^3}+\frac{6^{1/3}(B_++B_-)\gamma^{5/3}}{y_\star^{10/3}}\right)
\]
and, if \(g\) satisfies the ODE \eqref{eq: deq y} at \(y_-\), we get 
\[
D_+=D_-+\frac{3}{4}(M_+-M_-).
\]
If we insist on \(g'( y _-)=0\) at third order, we get 
\[
M_+=M_-=M_0 + \frac{6 y_\star ^{1/3}\gamma + 
6^{1/3} (5 B_- - 3 B_+) \gamma^{5/3}}{12 y_\star^{10/3}}.
\]
Next, we insist that the ODE is satisfied at any \(y\) in the interior of the no-trade region \([y_-,y_+]\), that is for any
\[
y=\theta y_-+(1-\theta)y_+,\quad 0<\theta<1.
\]
\(g\) satisfies the ODE at second order, only if 
\[
B_- = \left(\frac{\gamma y_\star}{6}\right)^{1/3}\frac{\mu}{\gamma \sigma^2}.
\]
We set \(B_+=B_-\), anticipating the same asymmetry in the second order expansion, as in other problems\footnote{See the situation for Levered ETFs in the previous section, or local-mean variance criteria \cite{em_1_2024}.}
Thus, in the interior of the interval, \(g\) satisfies the ODE at third order, only if \(E_+=E_-\), and \(E_+, L_-,L_+, C_x=(C_1+C_2)/4\) satisfy a system of four
linear equations; its Jacobian is non-singular with determinant
\[
29160\frac{y_\star^{12}}{\gamma^2}\neq 0,
\]
thus leads to a unique solution. The unique solution of this system ensures that the free boundary problem is satisfied at third order, which finishes this independent
proof of Lemma \ref{lem: y asy}. 
\end{proof}
\subsection{Risk and Returns}\label{sec: risk}
\subsubsection{Power Utility}\label{sec: power utility}
Long-run power utility investors with risk aversion $1\neq \gamma>0$ maximise the equivalent safe rate
\[
\lim\inf_{T\rightarrow \infty}\log \mathbb E[w_T (\varphi)^{1-\gamma}]^{\frac{1}{1-\gamma}}
\]
over all admissible strategies $\varphi\in\Phi$. The shadow price, which yields an optimal trading policy, is of the form \cite[Lemma B.2]{gerhold2012asymptotics}
\[
\widetilde S_t=g(\pi_t) S_t
\]
with
\[
g(\pi(\zeta))=\frac{w(\log(\zeta/\zeta_-))}{\zeta\left(1-w(\log(\zeta/\zeta_-))\right)},
\]
where $w$ satisfies the ODE
\[
\frac{1}{2} w''(y)+(1-\gamma) w'(y) w(y)+\left(\frac{\mu}{\sigma^2}-\frac{1}{2}\right) w'(y)=0.
\]
As 
\[
w(\log(\zeta/\zeta_-))=\frac{g(\pi(\zeta))\zeta }{1+g(\pi(\zeta)\zeta},
\]
$g$ satisfies the free boundary problem
\begin{align*}
&z (z-1)^2 \left((z-1) z g''(z)+2 g'(z) \left(-\pi \gamma +z^2 g'(z)+z-1\right)\right)\\&+(z-1) g(z) \left(2 \pi \gamma
-\left((z-1) z^3 g''(z)\right)-2 z^2 (-\pi \gamma +\gamma +z) g'(z)\right)-2 (\pi -1) \gamma z g(z)^2=0
\end{align*}
subject to \eqref{eq: shadow constraints}--\eqref{eq: smooth pasting}. We are again using the Ansatz
\eqref{eq: proxy g} for a polynomial approximation of $g$, and, inspired by asymptotics in \cite{em_1_2024} make the Ansatz for the trading boundaries
\begin{align}\label{eq: trading kappa boundaries}
\pi_\pm&:=\pi_*\pm\left(\frac{3}{4\gamma}\pi_*^2(\pi_*-1)^2\right)^{1/3}\varepsilon^{1/3}\\\nonumber&\qquad+\kappa_\pm\times \frac{(\pi_*)^2(1-\pi_*)}{6} \left(\frac{6}{\gamma\pi_*(1-\pi_*)}\right)^{2/3}\varepsilon^{2/3}+O(\varepsilon)
\end{align}
with unknown constants $\kappa_\pm$. By construction, $g(\pi_-)=1$ and $g(\pi_+)=1-\varepsilon$ at third order. To satisfy $g'( \pi_+)=0$ at second order, we must set
\[
G_+=G_-
\]
and
\[
C_-=G_--\frac{\gamma}{6 (\pi-1)^2\pi^2}
\]
and for having $g'( \pi_-)=0$ up to second order, we need to insist on $C_+=C_-$.

Next, if $g$ satisfies the ODE \eqref{eq: fbp psi1} at $\pi_\pm$ at second order, we get the coefficients $D_\pm$
(lengthy expressions in terms of $\kappa_\pm$, $M_0$, $M_\pm$) and they satisfy 
\[
D_+-D_-=\frac{3}{4}(M_+-M_-).
\]
If we insist on $g'(\pi_\pm)=0$ at third order, we get $M_\pm$ as functions of
$M_\mp$, $M_0$, and $\kappa_\pm$. 

Note that at the moment, the ODE is satisfied at
$ \pi_\pm$ at second order only. We finally insist on the ODE be satisfied at any $z$ in the interval, that is
\[
z=\theta \pi_-+(1-\theta) \pi_+,\quad 0<\theta<1.
\]
In the interior of the interval, $g$ satisfies the ODE only if 
\[
\kappa_\pm=0.
\]
In other words, the asymptotics of the trading boundaries lack the second order terms, and thus are of the form
\begin{equation}\label{eq: trafosxy}
\pi_\pm=\pi_*\pm\left(\frac{3}{4\gamma}\pi_*^2(\pi_*-1)^2\right)^{1/3}\varepsilon^{1/3} +O(\varepsilon).
\end{equation}
\begin{remark}\label{rem: myopic}
Thus we have an alternative proof of \cite[eq. (2.9) of Theorem 2.2]{gerhold2012asymptotics}. Note that this is different to the two strategies for the myopic investor, where both the shadow market strategy, which is suboptimal,
and the optimal portfolio, have non-vanishing second order terms. In fact, the optimal strategy for the myopic investor in \cite{gm2020} satisfies $\kappa_\pm=\gamma-1$, and the Shadow market strategy
satisfies $\kappa_\pm=-(\gamma-1)$, and is suboptimal (cf. \cite[Theorem 3.3]{em_1_2024}, unless $\gamma=1$. 
\end{remark}

\begin{remark}\label{rem: 3criteria ATC}
For power utility-investors, the trading boundaries are of the form \eqref{eq: trafosxy} implying that the second terms even vanish ($B_+=B_-=0$, cf. ~Remark \ref{rem: second order same sign impact no}). The trading boundaries for local mean-variance investors are similar, with $B_+=B_-$, and also for the associated asymptotic shadow we have the same coefficients of the second order,
albeit of different sign, cf. \cite{em_1_2024}. Therefore, for these three strategies, the average trading costs agree for small bid-ask spread $\varepsilon$ up to third order, and are of the form
\[
\text{ATC}=\frac{3 \sigma^2}{\gamma}\left(\frac{\gamma\pi_*(\pi_*-1)}{6}\right)^{4/3}\varepsilon^{2/3}-\frac{\mu(\gamma-1)}{2\gamma} \pi_*(\pi_*-1) \varepsilon+ O(\varepsilon^{4/3}).
\]
\end{remark}

\subsubsection{Logarithmic Utility}\label{sec: log utility}
A long-run log-utility investor maximizes logarithmic utility of terminal wealth for an infinite time horizon,
\begin{equation}\label{eq: log criterion}
\lim\sup_{T\rightarrow\infty}\frac{1}{T}\mathbb E[\log(w_T(\varphi))]
\end{equation}
over all admissible strategies $\varphi\in\Phi$. For $\mu/\sigma^2<1$ an optimal strategy
is derived in \cite{taksar1988diffusion}, with asymptotics established in \cite{gerhold2012asymptotics}, and, independently
by \cite{gm2020} (the latter discusses, more generally, a local mean-variance objective for any risk-aversion $\gamma\geq 0$, which agrees with log-utility for $\gamma=1$, and includes leverage, in particular $\mu/\sigma^2\geq 1$). In the parameterization of
\cite{em_1_2024}, the shadow price put forward by \cite{gerhold2012asymptotics} is
\[
\widetilde S_t=g(\pi_t) S_t,
\]
where $g, \pi_\pm$ solve 
\begin{align*}
\frac{1}{2}g''(\pi)\pi^2(1-\pi)^2 \sigma^2&=\frac{ \pi \sigma^2(g+g'(\pi)\pi(1-\pi))^2}{1-\pi+\pi g(\pi)}-\mu g(\pi)\\
&-g'(\pi)(\pi(1-\pi)\mu +\pi(1-\pi)^2\sigma^2)
\end{align*} 
subject to \eqref{eq: shadow constraints}--\eqref{eq: smooth pasting}. Note that this equation can also be obtained from the value function, more precisely \cite[Theorem 3.1, eqs. (3.1)-(3.5)]{gm2020} with the transformation \eqref{eq: trafo}.

Using the Ansatz \eqref{eq: proxy g} for a polynomial approximation of $g$, and the Ansatz for the trading boundaries
\begin{align*}
\pi_\pm&:=\pi_*\pm\left(\frac{3}{4}\pi_*^2(\pi_*-1)^2\right)^{1/3}\varepsilon^{1/3}\\&\qquad+\kappa_\pm\times \frac{(\pi_*)^2(1-\pi_*)}{6} \left(\frac{6}{\pi_*(1-\pi_*)}\right)^{2/3}\varepsilon^{2/3}+O(\varepsilon)
\end{align*}
with unknown constants $\kappa_\pm$. Now the same computations as in Section \ref{sec: power utility} deliver $\kappa_\pm=0$, and all constants in the approximation \eqref{eq: proxy g} of $g$ (they are as in Section \ref{sec: power utility}, with the exception that one needs to set the specific value $\gamma=1$.)

\subsubsection{Risk Neutrality}\label{sec: risk neutral}
By \cite{gm2020}, maximizing long-run returns, ignorant to risk, 
\[
\lim_{T\rightarrow \infty}\frac{1}{T}\mathbb E\left[\int_0^T\frac{dw_t}{w_t}\right]
\]
over all admissible strategies $\varphi\in\Phi$, is well posed, because transaction costs take the role of a penalty that typically variance of returns does.

Despite trading costs, for small bid-ask spreads $\varepsilon\in (0,1)$ excessive leverage, around $1/\sqrt{\varepsilon}$, can be obtained. Due to
\cite[Theorem 3.2]{gm2023},
for small enough bid-ask spread $\varepsilon$ there exist some $\pi_-<\pi_+$ such that the policy that minimally buys at $\pi_- $ and minimally sells at $\pi_+ $ to keep the fraction of wealth invested in the risky asset between $\pi_-$ and $\pi_+$ is optimal. The long-run expected return of this policy is given by
\begin{equation}\label{eq: welfare riskneutral}
\lim_{T\rightarrow\infty}
\frac1 T \int_0^T \frac{dw_t}{w_t}
= r + \mu \pi_- 
\end{equation}
and the trading boundaries have the series expansions
\begin{align}\label{as multiplier piminus}
\pi_- =& (1-\kappa)\kappa^{1/2}\left(\frac{\mu}{\sigma^2}\right)^{1/2}\varepsilon^{-1/2} + O(\varepsilon^{1}),\\
\pi_+ =& \kappa^{1/2}\left(\frac{\mu}{\sigma^2}\right)^{1/2}\varepsilon^{-1/2} + O(1),
\end{align}
with $\kappa\approx 0.58$ being the unique zero of
\[
\frac{3}{2}\xi+\log(1-\xi)=0
\]
in the unit interval.

It is important to note that the trading boundaries specified in equation \eqref{as multiplier piminus} exhibit singularity in their expansion, with the leading term experiencing explosive growth. Due to this characteristic, these boundaries do not align with the approach we have employed for other optimization problems. Additionally, according to \cite[Theorem 3.10]{em_1_2024}, there is no shadow price for this particular problem; not even an asymptotically optimal one, as opposed to the situation with positive risk-aversion, where a shadow price exists being optimal up to third order \cite[Theorem 3.1 and Theorem 3.3]{em_1_2024}.

Let $\alpha=\mu/\sigma^2$, and $(W,\zeta_-,\zeta_+)$ with $\zeta_-<\zeta_+$ be the unique solution (according to \cite[Theorem 3.2]{gm2020} of the free boundary problem 
\begin{align}\label{eq: TKA fbp}
&\textstyle{\frac{1}{2}\zeta^2 W''(\zeta)+(\alpha+1)\zeta W'(\zeta)+\alpha W(\zeta)-\frac{\alpha}{(1+\zeta)^2}=0,}\\\label{initial0 TKA}
&\textstyle{W(\zeta_-)=0},\quad
\textstyle{W'(\zeta_-)=0,}\\\label{terminal0 TKA}
&\textstyle{W(\zeta_+)=\frac{\varepsilon}{(1+\zeta_+)(1+(1-\varepsilon)\zeta_+)},}\\
\label{terminal1 TKA}
&\textstyle{W'(\zeta_+)=\frac{\varepsilon(\varepsilon-2(1-\varepsilon)\zeta_+-2)}{(1+\zeta_+)^2(1+(1-\varepsilon)\zeta_+)^2}}.
\end{align}
We regularize the problem with the transformation
\[
u=\frac{-1-\zeta}{\delta}, \quad \delta=\sqrt{\varepsilon},\quad \Psi(\cdot)=W(\zeta(\cdot))
\]
to get the following free boundary problem for the scalar function $\Psi$ and boundaries $u_\pm$
\begin{align}\label{eq: delta2ODE}
&\frac{(1+u\delta)^2}{2}\Psi''(u)+(1+\alpha)\delta(1+\delta u)\Psi'(u)+\alpha \delta^2 \Psi(u)-\frac{\alpha}{u^2}=0,\\\label{cond1}
&\Psi(u_-)=0,\\\label{cond2}
&\Psi'(u_-)=0,\quad\\\label{cond3}
&\Psi(u_+)=\frac{1}{u_+(1-\delta^2)-\delta},\\\label{cond4}
&\Psi'(u_+)=\frac{\delta +2 \left(\delta ^2-1\right) u_+}{ u_+^2
\left(\delta +\left(\delta ^2-1\right) u_+\right)^2}.
\end{align}
Note that, by the chain rule,
\[
\Psi'(u)=(W(\zeta(u))'=W'(\zeta(u))\frac{d\zeta}{d u}=-\delta W'(\zeta(u)),
\]
which implies \eqref{cond4}. We now make the Ansatz for an Approximation of the solution of \eqref{eq: delta2ODE},
\begin{equation}\label{eq: psi}
\Psi(u):=-2 \alpha \log(u)+C_1+C_2 u+\delta(C_3 u^{-\alpha}+C_4)+\delta^2(\frac{C_5}{u}+\frac{C_6}{u^{2\alpha}}).
\end{equation}
We have the following
\begin{lemma}
$\Psi$ defined by \eqref{eq: psi} is a solution of ODE \eqref{eq: delta2ODE} for any $C_i\in\mathbb R$ ($1\leq i \leq 6$).
\end{lemma} 
\begin{proof}
We expand \eqref{eq: delta2ODE} formally in powers of $\delta$ (only up to second order appear as coefficients), giving us the equation associated with $\delta^2$,
\begin{equation}\label{eq: initial}
\frac{u^2}{2}\Psi_2''(u)+(1+\alpha)u\Psi_2'(u)+\alpha\Psi_2(u)=0,
\end{equation}
the linear term associated with $\delta$
\begin{equation}\label{eq: mid}
u\Psi_1''(u)+(1+\alpha)\Psi_1'(u)=0,
\end{equation}
and the constant term
\begin{equation}
\frac{1}{2}\Psi_0''(u)-\frac{\alpha}{u^2}=0.
\end{equation}
The last one gives the general solution $\Psi_0(u)=-2 \alpha \log(u)+C_1+C_2 u$,
and equation \eqref{eq: mid} gives $\Psi_1(u)=C_3 u^{-\alpha}+C_4$. Finally,
equation \eqref{eq: initial} gives
\[
\Psi_2(u)=\frac{C_5}{u}+\frac{C_6}{u^{2\alpha}}.
\]
\end{proof}
Assume that
\[
\pi_\pm=\frac{1}{A_\pm \delta}+B_\pm+C_\pm \delta +O(\delta^2),
\]
then
\begin{equation}\label{eq zeta expansion}
\zeta_\pm=-1-A_\pm \delta+A_\pm^2 (B_\pm-1) \delta^2+A_\pm^2(C_\pm-A_\pm(B_\pm-1)^2)\delta^3+O(\delta^4)
\end{equation}
and thus
\[
u\pm=A_\pm-A_\pm^2(B_\pm-1)\delta-A_\pm^2(C_\pm-A_\pm(B_\pm-1)^2)\delta^2+O(\delta^3).
\]
We first make sure that the zero order term $\Psi_0$ satisfies, by default, the left boundary conditions \eqref{cond1} and \eqref{cond2}. This implies that
\[
\Psi_0(u)=-2\alpha \log(u/u_-)+2\alpha \frac{u-u_-}{u_-}.
\]

For condition \eqref{cond2} to be satisfied at first order, we get $C_3=0$. But then, to get \eqref{cond1} satisfied
at first order, we need $C_4=0$. We further set $C_5=C_6=0$ to satisfy the same conditions at second order. We thus have
\[
\Psi(u_-)=\Psi'(u_-)=O(\delta^3).
\]
Next, to have \eqref{cond3} and \eqref{cond4} satisfied at zero order, $A_\pm$ must satisfy the system of equations
\begin{equation}\label{first order eq1}
2\alpha\left(\log(A_-/A_+)-\frac{A_--A_+}{A_-}\right)-\frac{1}{A_+^2}=0,\quad
\alpha\left(\frac{1}{A_+}-\frac{1}{A_-}\right)-\frac{1}{A_+^3}=0.
\end{equation}
Note that these equations recover the ones in \cite[D.5]{gm2020}. By \cite[Lemma D.1.]{gm2020} the unique solution $(A_-,A_+)$ of the system \eqref{first order eq1} is
\begin{equation}\label{eq: first order proxies}
A_-=\frac{\kappa^{-1/2}}{1-\kappa}\sqrt{\frac{\sigma^2}{\mu}},\quad A_+=\kappa^{-1/2}\sqrt{\frac{\sigma^2}{\mu}},
\end{equation}
where $\kappa\approx 0.5828$ is the unique solution of \eqref{first order eq1}.

To have \eqref{cond3} and \eqref{cond4} satisfied at first order, we equate the respective coefficients to zero,
\begin{align}
\frac{-2 A_+^2 (B_+-1)-1}{A_+^3}-\frac{2 \alpha 
({A_-}-{A_+}) ({A_-} ({B_-}-1)-{A_+}
{B_+}+{A_+})}{{A_-}}&=0,\\
\frac{3}{{A_+}^4}+\frac{6 ({B_+}-1)}{{A_+}^2}+2 \alpha
({B_-}-{B-})&=0,
\end{align}
which yields the constants\footnote{The paper \cite{gm2020} states that $B_\pm=1$ which is a typo, having no implications for their paper.}
\begin{align}
B_-&=1+\frac{2 {A_-}-3 {A_+}}{2 \alpha
{A_+}^3 ({A_-}-{A_+})^2+2 {A_-} {A_+} (4
{A_+}-3 {A_-})},\\
B_+&=1+\frac{{A_-} (3
{A_-}-4 {A_+})}{2 \alpha {A_+}^4
({A_-}-{A_+})^2+2 {A_-} {A_+}^2 (4 {A_+}-3
{A_-})}.
\end{align}
We must ultimately verify that the approximation of the value function and trading strategy has been carried out to a satisfactory order. To achieve this, we begin by determining the order of average transaction costs. Since the trading boundaries experience explosive behavior as the bid-ask spread approaches zero, we anticipate a corresponding escalation in transaction costs. This expectation is corroborated by the following analysis: 
\begin{lemma}
The risk-neutral investor faces average transaction costs 
\begin{equation}\label{eq: ATCx}
\text{ATC}=\frac{\sigma^2 (1-\kappa) \kappa^{3/2}}{2\varepsilon^{1/2}}+O(1).
\end{equation}
\end{lemma}
\begin{proof}
The annualized mean trading costs for the strategy
with maximum returns is given by
\begin{equation}\label{eq: ATC appendix}
\avtrco:=\varepsilon\lim_{T\rightarrow\infty}\frac{1}{T}\int_0^T \pi_t \frac{d\varphi_t^\downarrow}{\varphi_t}=\frac{\sigma^2(2\alpha-1)}{2} \left(\frac{G({\zeta_+}) {\zeta_+}}{1-\left(\frac{{\zeta_-}}{{\zeta_+}}\right)^{2\gamma\pi_*-1}}\right),
\end{equation}
(cf.~By \cite[Lemma C.3]{gm2020}). Expanding this formula asymptotically, by using the expression \eqref{eq zeta expansion} for the free boundaries $\zeta_\pm$, leads to
\begin{equation}\label{eq: ATC0xx}
\text{ATC}=\frac{\sigma^2}{2(A_--A_+)A_+^2\varepsilon^{1/2}}+O(1)
\end{equation}
and thus, in view of \eqref{eq: first order proxies} yields \eqref{eq: ATCx}.
\end{proof}
Since by \eqref{eq: welfare riskneutral} and \eqref{as multiplier piminus} the equivalent safe rate adheres to the asymptotics
\[
r+\mu \pi_-= r+\frac{\mu^{3/2}\kappa^{1/2} (1-\kappa)}{\sigma\varepsilon^{1/2}}+O(1),
\]
we infer that the impact of transaction costs on welfare is of a similar order. Notably, this situation contrasts with many other problems, where the size of the no-trade region is strictly of larger order than the impact of transaction costs on welfare (as discussed in Section \ref{sec: Intro} and referenced therein). As the coefficients $B_\pm$ only contribute to the zeroth order, as seen in \eqref{eq: ATC0xx}, our approximation is deemed satisfactory. This suggests that the essential effects of transaction costs on welfare are adequately captured by the simplified analysis.

\section{Discussion}\label{sec: Discussion}

In this paper, we provide asymptotic methods for continuous-time portfolio choice under proportional transaction costs.

The first results pertain to trading frequency and incurred cost. While for long-run or infinite horizon optimisation problems, it is typically optimal to trade such that several portfolio statistics follow a reflected diffusion, these strategies cannot be realized in real-world trading setup: The trading is of infinite activity and thus requires not only continuous monitoring, but also the execution of trades of infinitesimal size. This inspires us to develop realistic, discrete trades. In terms of the proportion of wealth in the risky asset, if the trading boundaries satisfy
\[
\vert\pi_+-\pi_-\vert=O(\varepsilon^{1/3}),
\]
bulk trades to some target $\Lambda\in (\pi_-,\pi_+)$ incur average transaction costs that are twice than the optimal strategy's (Proposition \ref{prop: bulky trades} vs. Proposition \ref{prop: const prop minimal trades}), rendering them strictly suboptimal. For maximal trades (Section \ref{sec: maximal trades} and Remark \ref{rem: moderate trades}), where reasonable trade sizes $\Delta \pi$ must be of order $\varepsilon^{1/3}$ (no larger than the size of the no-trade region), the Trading frequency (TF) is of order $\varepsilon^{-2/3}$. A moment's reflection
on Proposition \ref{prop: bulky trades} reveals that annualized transaction costs are, precise at first order, the product of relative bid-ask spread, trading frequency and trade size:
\begin{equation}\label{eq: scaling law}
\text{ATC}=\varepsilon\times \text{TF}\times \Delta \pi.
\end{equation}
The same formula can be inferred from small trades (Proposition \ref{prop: micro trades}), where trading frequency is much higher (order $\varepsilon^{-1}$)
but annualized transaction costs are smaller (of the order $\varepsilon^{2/3}$). This process can be iterated, yielding the scaling law \eqref{eq: scaling law} (cf.~\cite{golub} and the references therein), for trade sizes of order $\varepsilon^{\alpha/3}$ and associated trading frequency $\varepsilon^{-\frac{1}{3}(\alpha+1)}$, where $\alpha=1,2,3,\dots$, resulting in the same order of average trading costs, namely $\varepsilon^{2/3}$. This scaling law provides a complimentary answer to the question in \cite{rogers} concerning the effect of proportional transaction costs being of second order, as annualized transaction costs of nearly optimal strategies are precisely of this order. For $\alpha\uparrow\infty$, we recover the control limit policy, which is continuous. Our result is applicable to any control limit policy, irrespective of the objective in use, whereas \cite{rogers} reaches his conclusion only for maximizing power utility from consumption in an investment problem on an infinite horizon.

Section \ref{sec: transaction costs} offers a novel insight into the question, why ATC may agree up to third (!) order, despite varying, problem-specific, second order terms in the trading boundaries (Remark \ref{rem: second order same sign impact no}): If the second order coefficients of lower and upper trading boundaries $\pi_\pm$ agree, they do not contribute to the asymptotics \eqref{eq: ATC GBM asy} of ATC, as only their difference, here zero, gives a non-zero contribution. A notable example is the (candidate) shadow price strategy in \eqref{tata} vs the optimal strategy \eqref{as multiplier piminust} for levered ETFs: The respective no-trade regions are small shifts - in opposed directions- of a symmetric region with boundaries
\[
\pi_\pm=\Lambda\pm\left(\frac{3}{4\gamma}\Lambda^2(\Lambda-1)^2\right)^{1/3}\varepsilon^{1/3}+O(\varepsilon).
\]
Note that the latter trading boundaries do not correspond to any optimisation problem mentioned here or in the literature (they lack second order terms.) However, in the analogous situation for local-mean variance criteria \cite{em_1_2024}, a symmetric no-trade region around the Merton fraction $\Lambda=\frac{\mu}{\gamma\sigma^2}$ identifies the optimal strategy for a power-utility investor \cite{gerhold.al.11} (see Remark \ref{rem: 3criteria ATC}). In this case, three different problems whose trading policy disagrees at second order, enjoy the same equivalent safe rate up to third order.

In Sections \ref{sec: tracking} and \ref{sec: risk}, we introduce a novel approximation approach for addressing free boundary problems. Employing a universal polynomial Ansatz for the value function enables the resolution of free boundary problems associated with power utility of terminal wealth, Levered ETFs objectives, and the computation of a solvent hedge for the Log contract. Remarkably, our solutions exhibit precision up to the third order, with trading boundaries precision reaching the second order. This high level of precision becomes particularly relevant, given that transaction costs' primary contribution to the objective is of the second order. These costs are solely influenced by the first-order coefficients of the trading boundaries. Consequently, the polynomial method excels in providing precision beyond what is strictly necessary.

\theendnotes
\appendix

\end{document}